\documentclass[letterpaper,11pt]{article}

\usepackage{amsmath}
\usepackage{amssymb}
\usepackage{amsthm}
\usepackage{xcolor}
\usepackage{complexity}
\usepackage{fullpage}
\usepackage[whole]{bxcjkjatype}
\usepackage{microtype}
\usepackage{enumitem}
\usepackage{thm-restate}
\usepackage{epigraph}
\usepackage{authblk}
\usepackage{hyperref}
\usepackage{cleveref}
\usepackage[
  backend=biber,
  style=alphabetic,
  natbib=false,
  maxnames=99,
  maxalphanames=4,
  minalphanames=3,
  useprefix=true
]{biblatex}

\hypersetup{
  linkcolor={red!75!black},
  citecolor={blue!75!black},
  urlcolor={red!75!black},
}

\theoremstyle{plain}
\newtheorem{theorem}{Theorem}[section]
\newtheorem{lemma}[theorem]{Lemma}
\newtheorem{proposition}[theorem]{Proposition}
\newtheorem{fact}[theorem]{Fact}

\theoremstyle{definition}
\newtheorem{definition}[theorem]{Definition}

\theoremstyle{remark}
\newtheorem{remark}[theorem]{Remark}

\newcommand{\Z}{\mathbb{Z}}
\newcommand{\F}{\mathbb{F}}
\newcommand{\card}[1]{\left| {#1} \right|}
\renewcommand{\vec}[1]{\ensuremath{\boldsymbol{#1}}}
\newcommand{\problem}[1]{\ensuremath{\mathrm{#1}}}
\newcommand{\lat}{\ensuremath{\mathcal{L}}}
\newcommand{\GapSVP}{\problem{GapSVP}}
\newcommand{\GapCVP}{\problem{GapCVP}}
\providecommand{\GapNCP}{\problem{GapNCP}}
\providecommand{\GapMDP}{\problem{GapMDP}}
\providecommand{\OR}{\problem{OR}}

\renewcommand{\C}{\mathcal{C}}

\usepackage[draft,multiuser,inline,nomargin]{fixme}
\usepackage[left=1in, right=1in, top=1in]{geometry}
\begin{document}

\title{One-Sided-Error Parameterized Reductions for \\the Minimum Distance and Shortest Vector Problems}

\author[1]{Shuichi Hirahara}
\author[2]{Kazuki Ogitsuka}
\affil[1]{National Institute of Informatics, Japan \protect\\ \texttt{\href{mailto:s_hirahara@nii.ac.jp}{s\_hirahara@nii.ac.jp}}}
\affil[2]{The Graduate University for Advanced Studies, SOKENDAI \protect\\ \texttt{\href{mailto:ogitsuka@nii.ac.jp}{ogitsuka@nii.ac.jp}}}

\date{}
\maketitle

\begin{abstract}
It is notoriously difficult to obtain deterministic reductions for the Minimum Distance Problem (MDP) and the Shortest Vector Problem (SVP).
Under two-sided-error randomized reductions,
Bennett, Cheraghchi, Guruswami, and Ribeiro (STOC 2023) proved parameterized hardness of approximation for these problems.
We partially derandomize their reductions and present one-sided-error randomized reductions:
MDP is $\W[1]$-hard to approximate within an arbitrary constant factor under $\FPT$ many-one one-sided-error randomized reductions;
For every $p \ge 1$,
SVP in the $\ell_p$ norm is $\W[1]$-hard to approximate within an arbitrary constant factor below $2^{1/p}$.

We demonstrate the usefulness of one-sided-error randomized reductions by showing that they can be conditionally derandomized when the target problem has an OR function.
Under a standard hardness-vs-randomness assumption, namely a plausible lower-bound assumption against nondeterministic circuits, we prove a general theorem formalizing this derandomization.
Here, an OR function combines several instances into one instance that preserves their disjunction.
We construct such OR functions for the relevant MDP and SVP gap problems, and thereby obtain deterministic $\W[1]$-hardness for approximating MDP over every fixed finite field within every constant factor, and for approximating SVP in $\ell_p$ norms for every fixed integer $p$ within every factor below $2^{1/p}$.
Applying the same framework to Micciancio's one-sided-error randomized reduction (ToC 2012) yields, under the same circuit lower-bound assumption, deterministic polynomial-time $\NP$-hardness of approximating Euclidean SVP within every constant factor.
\end{abstract}

\newpage
\section{Introduction}
As fundamental computational problems on codes and lattices, the Minimum Distance Problem (MDP) and the Shortest Vector Problem (SVP) have long been studied.
MDP asks for the minimum Hamming weight of a nonzero codeword in a given linear code, and can be viewed as the homogeneous analogue of the Nearest Codeword Problem (NCP).
Similarly, SVP asks for the shortest nonzero vector in a given integer lattice, and can be viewed as the homogeneous analogue of the Closest Vector Problem (CVP).
Both problems play central roles in coding theory, lattice algorithms, cryptography, and computational complexity theory.
For further details, we refer the reader to surveys and standard textbooks \cite{survey/AgrellEVZ02,book/Lin_Li_2021,survey/Peikert16,survey/weger2024surveycodebasedcryptography}.

In approximation hardness results, it is standard to consider the corresponding gap versions.
For an approximation factor \(\gamma \ge 1\), \(\gamma\)-\(\GapMDP_q\) over a finite field \(\mathbb{F}_q\) asks us to distinguish between the case where the given code contains a nonzero codeword of Hamming weight at most \(k\), and the case where every nonzero codeword has Hamming weight greater than \(\gamma k\).
Similarly, \(\gamma\)-\(\GapSVP_p\) asks us to distinguish between the case where the given lattice contains a nonzero vector of \(\ell_p\) norm at most \(d\), and the case where every nonzero lattice vector has \(\ell_p\) norm greater than \(\gamma d\).

\subparagraph*{Classical complexity of MDP and SVP.}
In the classical setting of computational complexity, there is a large body of work on the hardness of approximating MDP and SVP.
For MDP, Dumer, Micciancio, and Sudan \cite{journals/DumerMS03} proved the \(\NP\)-hardness of \(\gamma\)-\(\GapMDP_q\) for every constant \(\gamma \ge 1\) under randomized reductions.
This randomized hardness result was later derandomized by Cheng and Wan \cite{journals/IEEE/ChengW12}.
Austrin and Khot \cite{journals/AustrinK14} gave a simpler deterministic reduction, and Micciancio \cite{conf/ccc/Micciancio14} further simplified the construction.

For SVP, van Emde Boas \cite{vanEmdeBoas} proved the \(\NP\)-hardness of the decision version of SVP in the \(\ell_\infty\) norm, which corresponds to the case \(\gamma=1\) in the gap formulation above, and conjectured that SVP is \(\NP\)-hard also in the Euclidean \((\ell_2\)) norm.
Ajtai \cite{conf/stoc/Ajtai98} affirmatively resolved this conjecture under randomized reductions by proving the \(\NP\)-hardness of the decision version of SVP in the \(\ell_2\) norm.
Subsequent work by Khot \cite{journals/JACM/Kho05,journals/jcss/Kho06}, Haviv and Regev \cite{journals/toc/HR12}, and Micciancio \cite{journals/sicomp/Micciancio01,journals/toc/Micciancio12} established stronger \(\NP\)-hardness results for \(\gamma\)-\(\GapSVP_p\) for larger approximation factors and various \(\ell_p\) norms under randomized reductions.
Bennett and Peikert \cite{conf/random/BennettP23} gave a simple alternative proof of Micciancio's \(\NP\)-hardness result \cite{journals/sicomp/Micciancio01} for \(\gamma\)-\(\GapSVP_p\) at approximation factors \(\gamma \in [1, 2^{1/p})\) under randomized reductions, and their proof appears more suitable for derandomization.

In very recent breakthroughs, Hecht and Safra \cite{conf/stoc/HechtS26} and Hair and Sahai \cite{conf/stoc/HairS26} proved hardness of approximation for \(\GapSVP_p\) for every \(p>2\) under deterministic reductions.
Subsequently, Hair and Sahai \cite{hair2026finite} extended this hardness to all finite \(\ell_p\) norms under deterministic subexponential-time reductions.
Wan \cite{wan2026} resolved van Emde Boas's long-standing conjecture by proving the \(\NP\)-hardness of the decision version of SVP in the \(\ell_2\) norm under deterministic polynomial-time reductions.
More generally, \cite{wan2026} proved the \(\NP\)-hardness of \(\gamma\)-\(\GapSVP_p\) for approximation factors \(\gamma \in [1,2^{1/p})\) under deterministic polynomial-time reductions for every fixed \(p \ge 1\).

\subparagraph*{Parameterized complexity of MDP and SVP.}
\emph{Parameterized complexity} emerged as a systematic framework in the late 1980s and early 1990s, mainly through the work of Downey and Fellows \cite{journals/sicomp/DF95a,journals/tcs/DF95b,book/DowneyF99}, building on earlier work with Langston; see also \cite{book/CyganFKLMPPS15}.
In this framework, the running time of an algorithm is measured not only as a function of the input length \(n\), but also as a function of an additional parameter \(k\).
A parameterized problem is fixed-parameter tractable if it can be solved in time \(f(k)\cdot n^{O(1)}\) for some computable function \(f\).
The class of such problems is denoted by \(\FPT\).
The standard hardness notion used to rule out fixed-parameter tractability is \(\W[1]\)-hardness, which plays a role analogous to \(\NP\)-hardness in classical complexity theory.

For \(\gamma\)-\(\GapMDP_q\) and \(\gamma\)-\(\GapSVP_p\), one of the natural parameters is the distance threshold.
The binary case of parameterized MDP is closely related to the \(k\)-Even Set problem, which asks whether a binary linear code contains a nonzero codeword of Hamming weight at most \(k\).
This problem had long been a central open problem in parameterized complexity.
Bhattacharyya, Bonnet, Egri, Ghoshal, Karthik C. S., Lin, Manurangsi, and Marx \cite{journals/BhattacharyyaBE+21} made major progress by proving that, for every constant \(\gamma \ge 1\), \(\gamma\)-\(\GapMDP_2\) is \(\W[1]\)-hard under randomized reductions.
They also proved that, for every fixed \(p>1\), \(\gamma\)-\(\GapSVP_p\) is \(\W[1]\)-hard for every \(\gamma \in [1,(1/2+1/2^p)^{-1/p})\), again under randomized reductions.

Bennett, Cheraghchi, Guruswami, and Ribeiro \cite{conf/stoc/BennettCGR23} later extended these parameterized inapproximability results.
They proved that, for every fixed finite field \(\mathbb{F}_q\) and every constant \(\gamma \ge 1\), \(\gamma\)-\(\GapMDP_q\) is \(\W[1]\)-hard under randomized reductions.
They also proved parameterized inapproximability results for \(\gamma\)-\(\GapSVP_p\) in all finite \(\ell_p\) norms.
More precisely, for every fixed \(p>1\), they proved \(\W[1]\)-hardness for every constant approximation factor \(\gamma \ge 1\).
For the remaining case \(p=1\), they proved \(\W[1]\)-hardness for every \(\gamma \in [1,2)\).

\subparagraph*{One-sided error vs two-sided error.}
For randomized many-one reductions, there is a distinction between \emph{two-sided error} and \emph{one-sided error}.
In a reduction with two-sided error, YES instances are mapped to YES instances with high probability, and NO instances are mapped to NO instances with high probability.
Thus, an unlucky choice of randomness may cause an error on either side.
In a reduction with one-sided error, errors are allowed only on one side.
In this paper, one-sided error means that YES instances are mapped to YES instances with high probability, whereas NO instances are always mapped to NO instances.

This distinction played an important role in Micciancio's work on the classical hardness of \(\gamma\)-\(\GapSVP_2\).
In particular, Micciancio \cite{journals/toc/Micciancio12} explicitly viewed randomized reductions with one-sided error as a \emph{partial derandomization}, and presented them as progress toward deterministic reductions.
Similar randomized reductions with one-sided error for \(\gamma\)-\(\GapSVP_p\) had already appeared in Micciancio's earlier work \cite{journals/sicomp/Micciancio01}, and were also used by Bennett and Peikert \cite{conf/random/BennettP23}.
Wan \cite{wan2026} subsequently derandomized this line of reductions, proving the \(\NP\)-hardness of \(\gamma\)-\(\GapSVP_p\) under deterministic polynomial-time reductions for every fixed \(p\ge 1\) and every constant \(\gamma \in [1,2^{1/p})\).

For parameterized MDP and SVP, one-sided error was previously known only in the binary MDP setting, while the broader parameterized hardness reductions remained two-sided.
Bhattacharyya, Bonnet, Egri, Ghoshal, Karthik C. S., Lin, Manurangsi, and Marx \cite{journals/BhattacharyyaBE+21} proved that, for every constant \(\gamma\ge 1\), \(\gamma\)-\(\GapMDP_2\) is \(\W[1]\)-hard under randomized reductions with one-sided error.
Their parameterized inapproximability result for \(\gamma\)-\(\GapSVP_p\), however, was obtained under randomized reductions with two-sided error.
The later reductions of Bennett, Cheraghchi, Guruswami, and Ribeiro~\cite{conf/stoc/BennettCGR23} also have two-sided error in all cases, both for \(\gamma\)-\(\GapMDP_q\) over any fixed finite field and for \(\gamma\)-\(\GapSVP_p\) in finite \(\ell_p\) norms.

\subsection{Our Results}

Our first contribution is to strengthen the parameterized hardness reductions of Bennett, Cheraghchi, Guruswami, and Ribeiro~\cite{conf/stoc/BennettCGR23} from two-sided error to one-sided error.

Our first theorem proves this for \(\gamma\)-\(\GapMDP_q\) over every fixed finite field, without losing their hardness \cite{conf/stoc/BennettCGR23} for every constant approximation factor.

\begin{theorem}
\label{thm:MDP-is-W[1]hard}
For every fixed prime power \(q \ge 2\) and every constant \(\gamma \ge 1\), \(\gamma\)-\(\GapMDP_q\) is \(\W[1]\)-hard under randomized \(\FPT\) many-one reductions with one-sided error.
\end{theorem}

Our second theorem proves the corresponding statement for \(\gamma\)-\(\GapSVP_p\) in the range \(\gamma<2^{1/p}\).

\begin{theorem}
\label{thm:SVP-is-W[1]hard}
For any fixed \(p \in [1,\infty)\) and constant \(\gamma \in [1,2^{1/p})\), \(\gamma\)-\(\GapSVP_p\) is \(\W[1]\)-hard under randomized \(\FPT\) many-one reductions with one-sided error.
\end{theorem}

In the classical setting, Micciancio~\cite{journals/toc/Micciancio12} viewed his randomized reduction with one-sided error for the \(\NP\)-hardness of \(\gamma\)-\(\GapSVP_2\) as a \emph{partial derandomization}.
Following this terminology, \Cref{thm:MDP-is-W[1]hard,thm:SVP-is-W[1]hard} can be viewed as a partial derandomization of the randomized parameterized hardness reductions of MDP and SVP.

The significance of one-sidedness is that it provides a route to removing the remaining randomness under standard hardness-vs-randomness assumptions.
Our second contribution is to show that, for MDP and SVP, the randomized hardness results with one-sided error above can be converted into deterministic hardness results under a standard circuit lower-bound assumption.
We assume that there exists an \(\varepsilon>0\) such that
\begin{align*}
    \E := \DTIME(2^{O(n)}) \nsubseteq i.o.\text{-}\mathrm{NSIZE}(2^{\varepsilon n}).
\end{align*}
Here, \(\mathrm{NSIZE}(s(n))\) denotes the class of languages decidable by nondeterministic circuits of size \(O(s(n))\), and \(\mathrm{i.o.}\) stands for ``infinitely often''.
\begin{theorem}
\label{thm:MDP-is-W[1]hard-conditional}
Assume that there exists an \(\varepsilon > 0\) such that
\begin{align*}
    \E \nsubseteq i.o.\text{-}\mathrm{NSIZE}(2^{\varepsilon n}).
\end{align*}
Then, for every fixed prime power \(q \ge 2\) and constant \(\gamma \ge 1\), \(\gamma\)-\(\GapMDP_q\) is \(\W[1]\)-hard under deterministic \(\FPT\) many-one reductions.
\end{theorem}

\begin{theorem}
\label{thm:SVP-is-W[1]hard-conditional}
Assume that there exists an \(\varepsilon > 0\) such that
\begin{align*}
    \E \nsubseteq i.o.\text{-}\mathrm{NSIZE}(2^{\varepsilon n}).
\end{align*}
For every fixed integer $p \ge 1$ and constant \(\gamma \in [1,2^{1/p})\), \(\gamma\)-\(\GapSVP_p\) is \(\W[1]\)-hard under deterministic \(\FPT\) many-one reductions.%
\footnote{
More generally, the result holds for every fixed \(p\ge 1\) such that the language of pairs \((\mathbf{B},d)\), where \(\mathbf{B}\in\mathbb{Z}^{m\times n}\) is a lattice basis, \(d\in\mathbb{Z}_{>0}\), and there exists \(\vec{z}\in\mathbb{Z}^n\setminus\{\vec{0}\}\) with \(\|\mathbf{B}\vec{z}\|_p\le d\), belongs to \(\NP\).
}
\end{theorem}

Finally, the conditional derandomization principle also applies to the classical, non-parameterized setting.
Starting from Micciancio's randomized hardness result with one-sided error for \(\gamma\)-\(\GapSVP_2\) \cite{journals/toc/Micciancio12}, we obtain a deterministic hardness result under the same circuit lower-bound assumption.

\begin{theorem}
\label{thm:SVP-is-NPhard-conditional}
Assume that there exists an \(\varepsilon > 0\) such that
\begin{align*}
    \E \nsubseteq i.o.\text{-}\mathrm{NSIZE}(2^{\varepsilon n}).
\end{align*}
Then, for every fixed constant \(\gamma \ge 1\), \(\gamma\)-\(\GapSVP_2\) is \(\NP\)-hard under deterministic polynomial-time many-one reductions.
\end{theorem}

Compared with Wan's unconditional deterministic polynomial-time many-one hardness result for \(1 \le \gamma < \sqrt{2}\)~\cite{wan2026}, \Cref{thm:SVP-is-NPhard-conditional} extends the approximation range to every constant \(\gamma \ge 1\), at the cost of the circuit lower-bound assumption above.

Under the assumption $\E \nsubseteq i.o.\text{-}\mathrm{NSIZE}(2^{\varepsilon n})$, we show in \Cref{sec:derandomize} that suitable one-sided-error randomized reductions to promise problems with OR functions can be made deterministic.
An OR function combines a finite list of instances into a single instance that is YES if at least one input is YES, and is NO if all inputs are NO.
In the parameterized setting, we additionally require the output parameter to be bounded by a computable function of the largest input parameter.

To apply this framework to MDP and SVP, we show that the relevant gap problems have the required OR functions.

\begin{proposition}[Informal; see \Cref{lem:gapmdp-fpt-or-function,lem:gapsvp-poly-or-function,lem:gapsvp-fpt-or-function}]
\label{prop:OR-functions-intr}
Let \(\gamma\ge 1\) be a constant.
The following hold.
\begin{enumerate}
    \item For every fixed prime power \(q\ge 2\), \(\gamma\)-\(\GapMDP_q\), parameterized by the distance threshold, has an FPT OR function.

    \item For every fixed \(p\ge 1\), \(\gamma\)-\(\GapSVP_p\) has a polynomial-time OR function.
    Moreover, when parameterized by the distance threshold, \(\gamma\)-\(\GapSVP_p\) has an FPT OR function.
\end{enumerate}
\end{proposition}
\subsection{Technical Overview}

Our reductions follow the same high-level route as those of Bennett, Cheraghchi, Guruswami, and Ribeiro \cite{conf/stoc/BennettCGR23}.
We start from the inhomogeneous problems NCP and CVP, whose parameterized hardness is known \cite{journals/BhattacharyyaBE+21}, and reduce them to the corresponding homogeneous problems MDP and SVP, respectively.
Here, NCP asks whether a target vector is close to some codeword of a given linear code, while CVP asks whether a target vector is close to some vector of a given lattice.

\paragraph*{The Ajtai--Micciancio approach and Khot's approach.}
We first recall two approaches for converting inhomogeneous problems on codes and lattices into homogeneous ones.
Let \((\mathbf{B},\vec{t})\) denote the inhomogeneous input instance, where \(\mathbf{B}\) is either the generator matrix of an NCP instance or the lattice basis of a CVP instance, and let \((\mathbf{A},\vec{s})\) denote a locally dense code or lattice.
Informally, \((\mathbf{A},\vec{s})\) has two useful properties: every nonzero vector generated by \(\mathbf{A}\) is long, while many shifted vectors of the form \(\mathbf{A}\vec{z}-\vec{s}\) are short.
Thus, in a YES instance, one can combine a short residual \(\mathbf{B}\vec{x}-\vec{t}\) with a nearby shifted vector, whereas in a NO instance, the large minimum distance of the locally dense object helps rule out unintended short vectors.

The Ajtai--Micciancio approach \cite{conf/stoc/Ajtai98,journals/sicomp/Micciancio01,journals/toc/Micciancio12,conf/ccc/BennettP20,wan2026} uses an additional map \(\mathbf{T}\) from coefficient vectors in the locally dense gadget to candidate solutions of the inhomogeneous instance.
In the simplified form relevant here, many binary vectors can be represented as
\begin{align*}
    \vec{x}=\mathbf{T}\vec{z}
    \quad\text{with}\quad
    \|\mathbf{A}\vec{z}-\vec{s}\|\le \alpha \ell .
\end{align*}
Given an NCP/CVP instance whose YES witness is a binary vector \(\vec{x}\in\{0,1\}^k\), one considers the homogeneous object generated by
\begin{align*}
    \begin{pmatrix}
        \mathbf{B}\mathbf{T} & -\vec{t} \\
        \mathbf{A}           & -\vec{s}
    \end{pmatrix}.
\end{align*}
If the input is a YES instance and \(\mathbf{T}\vec{z}=\vec{x}\), then setting the last coefficient to one gives the short vector
\begin{align*}
    (\mathbf{B}\vec{x}-\vec{t},\;
     \mathbf{A}\vec{z}-\vec{s}).
\end{align*}
The reduction may still be randomized because the locally dense gadget is sampled.
A good sample is needed for completeness, since it provides enough nearby shifted vectors.
In contrast, soundness follows from the minimum distance property and holds for every possible sample produced by the construction.
Hence, an unlucky sample can only destroy completeness; it cannot create a false YES instance from a NO input.
This gives one-sided error.

Although Wan~\cite{wan2026} recently made the locally dense lattice construction deterministic in the classical setting, the Ajtai--Micciancio approach still does not directly yield parameterized hardness.
The known \(\W[1]\)-hardness results for parameterized NCP and CVP are for general gap versions of these problems \cite{journals/BhattacharyyaBE+21}.
A YES instance only guarantees the existence of some coefficient vector that yields a codeword or lattice vector close to the target; it does not guarantee that this coefficient vector is binary or compatible with a fixed map \(\mathbf{T}\).
Moreover, an \(\FPT\) reduction must keep the output distance threshold bounded by a function of the input threshold alone.

The parameterized reductions of Bennett, Cheraghchi, Guruswami, and Ribeiro \cite{conf/stoc/BennettCGR23} instead follow Khot's \emph{sparsification} idea \cite{journals/JACM/Kho05}.
This approach is well suited to parameterized reductions because it does not require the YES witness to have a prescribed binary form, and it allows the output threshold to remain bounded in terms of the input threshold.

One first forms the intermediate homogeneous object generated by
\begin{align*}
    \mathbf{B}_{\mathrm{int}}
    =
    \begin{pmatrix}
        \mathbf{B} & \mathbf{0} & -\vec{t} \\
        \mathbf{0} & \mathbf{A} & -\vec{s}
    \end{pmatrix},
\end{align*}
where, for simplicity, we omit the extra row for the scalar coordinate that is often included in the lattice case.
If the input is a YES instance, then, by local density, this intermediate object contains many short vectors obtained by combining a short residual \(\mathbf{B}\vec{x}-\vec{t}\) with a nearby shifted vector \(\mathbf{A}\vec{y}-\vec{s}\).
However, the intermediate object may also contain unwanted short vectors, which are often called ``annoying'' vectors.
Sparsification is used to perform two tasks at once: it should preserve a good vector in YES instances, and eliminate all annoying vectors in NO instances.
In the coding case, this is done by intersecting the intermediate code with a random code.
In the lattice case, this is done by passing to a random sublattice, for example
\begin{align*}
    \mathcal{L}_{\mathrm{final}}
    =
    \{\vec{w}\in \mathcal{L}_{\mathrm{int}}:
      \langle \vec{v},\vec{w}\rangle=0 \pmod \rho\},
\end{align*}
where \(\mathcal{L}_{\mathrm{int}}\) denotes the lattice generated by the basis \(\mathbf{B}_{\mathrm{int}}\), and $\vec{v}$ is a random vector.
Thus, the resulting reductions have two-sided error: a bad random choice may fail either to preserve a good vector in a YES instance or to eliminate all annoying vectors in a NO instance.

\paragraph*{Our approach.}
The high-level idea is to keep Khot's intermediate homogeneous object, but to remove the annoying vectors deterministically instead of by random sparsification.
In the reductions of Bennett, Cheraghchi, Guruswami, and Ribeiro~\cite{conf/stoc/BennettCGR23}, the random sparsification both preserves a good vector in a YES instance and eliminates annoying vectors in a NO instance.
We keep randomness only for the first task.
For a NO instance, any candidate short vector is first forced into a very restricted form: the scalar coordinate must be zero, the locally dense component must vanish, and only a short nonzero first block can remain.
The deterministic checking block is designed to detect exactly this remaining case.
In the coding reduction, this is done by a BCH check followed by syndrome amplification, and in the lattice reduction it is done by a Reed--Solomon check modulo \(q\) followed by a large scaling factor.
Thus, every random choice is safe on NO inputs, while YES inputs still succeed with high probability.

For MDP, we start from an instance \((\mathbf{G},\vec{t},k)\) of NCP.
As in previous reductions, we use a locally dense code \((\mathbf{A},\vec{s})\) and a random linear map \(\mathbf{R}\).
In addition, we use a deterministic linear map \(\mathbf{P}\), obtained from a BCH parity-check matrix, that detects every nonzero vector of small Hamming weight.
We also use a linear map \(\mathbf{E}\) that amplifies nonzero syndromes so that
\begin{align*}
    \vec{c}\ne \vec{0}
    \quad\Longrightarrow\quad
    \|\mathbf{E}\vec{c}\|_0>d .
\end{align*}
The output code is generated by
\begin{align*}
    \mathbf{G}' =
    \begin{pmatrix}
        \mathbf{G}   & \mathbf{0}   & -\vec{t} \\
        \mathbf{0}   & \mathbf{A}   & -\vec{s} \\
        \mathbf{E}\mathbf{P}\mathbf{G}
        & \mathbf{E}\mathbf{R}\mathbf{A}
        & -\mathbf{E}(\mathbf{P}\vec{t}+\mathbf{R}\vec{s})
    \end{pmatrix}.
\end{align*}
Then the codeword generated by $\mathbf{G'}$ from the coefficient vector $(\vec{x}, \vec{y}, \beta)$ is  $\mathbf{G}'(\vec{x}, \vec{y}, \beta)^T=(\mathbf{G}\vec{x} - \beta\vec{t}, \mathbf{A}\vec{y}-\beta\vec{s}, \mathbf{E}(\mathbf{P}(\mathbf{G}\vec{x} - \beta\vec{t})+\mathbf{R}(\mathbf{A}\vec{y}-\beta\vec{s})))^T$.
The first two block rows form the usual intermediate object, and the last block row checks the syndrome \(\mathbf{P}\vec{u}+\mathbf{R}\vec{w}\), where \(\vec{u}=\mathbf{G}\vec{x}-\beta\vec{t}\) and \(\vec{w}=\mathbf{A}\vec{y}-\beta\vec{s}\).

In a YES instance, there is a short residual \(\vec{u}'=\mathbf{G}\vec{x}'-\vec{t}\).
The locally dense code supplies many nearby vectors \(\vec{w}=\mathbf{A}\vec{y}-\vec{s}\), and the random map \(\mathbf{R}\) makes it likely that one of them satisfies \(\mathbf{R}\vec{w}=-\mathbf{P}\vec{u}'\).
For such a vector, the checking block vanishes, and the output code contains a short nonzero codeword.

In a NO instance, the argument is independent of \(\mathbf{R}\).
If \(\beta\ne 0\), then the first block contradicts the NO promise of the NCP instance.
If \(\beta=0\) and the locally dense block is nonzero, then that block is a nonzero codeword of the locally dense code and is too heavy.
If the locally dense block vanishes, then every short nonzero first block is detected by \(\mathbf{P}\), and the amplified syndrome is too heavy.
Therefore, no short nonzero codeword exists in a NO instance, for every choice of \(\mathbf{R}\).

For SVP, the same idea is implemented modulo a prime \(q\).
We start from an instance \((\mathbf{B},\vec{t},d)\) of CVP.
We use a locally dense lattice \((\mathbf{A},\vec{s})\), a random linear map \(\mathbf{R}\) over \(\mathbb{F}_q\), and a deterministic linear map \(\mathbf{P}\) obtained from a Reed--Solomon parity-check matrix.
The matrix \(\mathbf{P}\) detects every nonzero vector with small support after reduction modulo \(q\), which is enough for soundness because every sufficiently short integer vector has small support and has all coordinates of absolute value less than \(q/2\).

We add a congruence version of the checking block to the output lattice.
The output lattice has basis
\begin{align*}
\mathbf{B}' =
\begin{pmatrix}
\mathbf{B}       & \mathbf{0}       & -\vec{t}              & \mathbf{0} \\
\mathbf{0}       & \mathbf{A}       & -\vec{s}              & \mathbf{0} \\
\vec{0}^{\mathsf T} & \vec{0}^{\mathsf T} & 1                 & \vec{0}^{\mathsf T} \\
\eta \widehat{\mathbf{P}}\mathbf{B}
        & \eta \widehat{\mathbf{R}}\mathbf{A}
        & -\eta(\widehat{\mathbf{P}}\vec{t}+\widehat{\mathbf{R}}\vec{s})
        & \eta q \mathbf{I}
\end{pmatrix}.
\end{align*}
Here, \(\widehat{\mathbf{P}}\) and \(\widehat{\mathbf{R}}\) are integer lifts of \(\mathbf{P}\) and \(\mathbf{R}\), and \(\eta\) is chosen to be larger than the target distance.
The last block row forces every sufficiently short vector to satisfy
\begin{align*}
    \widehat{\mathbf{P}}\vec{u}+\widehat{\mathbf{R}}\vec{w}+q\vec{z}=\vec{0},
\end{align*}
where \(\vec{u}=\mathbf{B}\vec{x}-a\vec{t}\) and \(\vec{w}=\mathbf{A}\vec{y}-a\vec{s}\).

In a YES instance, there is a short residual \(\vec{u}'=\mathbf{B}\vec{x}'-\vec{t}\).
The random map \(\mathbf{R}\) is used to find a nearby point \(\vec{w}=\mathbf{A}\vec{y}-\vec{s}\) satisfying $\mathbf{R}(\vec{w}\bmod q) = -\mathbf{P}(\vec{u}'\bmod q)$.
Once this congruence holds, one can choose \(\vec{z}\) so that the checking block vanishes, yielding a short output vector.

In a NO instance, the argument is again independent of \(\mathbf{R}\).
If \(a\ne 0\), then the first block contradicts the strengthened NO promise of the CVP instance, which rules out all nonzero integer multiples of the target.
If \(a=0\) and the locally dense block is nonzero, then that block is a nonzero vector of the locally dense lattice and is too long.
If both of these blocks vanish, then a nonzero output vector must either use the \(q\mathbf{I}\) block or contain a nonzero short vector \(\vec{u}\) in the first block.
In the first case, the factor \(\eta q\) makes the vector too long.
In the second case, \(\mathbf{P}(\vec{u}\bmod q)\ne \vec{0}\), so the congruence block cannot vanish, and the factor \(\eta\) makes it too long.
Therefore, no short nonzero lattice vector exists in a NO instance, for every choice of \(\mathbf{R}\).
\paragraph*{Conditional hardness.}
We next explain how one-sidedness leads to deterministic reductions under a circuit lower-bound assumption.
The general argument is not specific to codes or lattices.
It applies to any randomized reduction with one-sided error whose successful random tapes can be recognized by small nonadaptive \(\SAT\)-oracle circuits, provided that the target problem admits a way to combine several candidate outputs into a single instance.

We assume that there exists an \(\varepsilon>0\) such that
\begin{align*}
    \E \nsubseteq i.o.\text{-}\mathrm{NSIZE}(2^{\varepsilon n}).
\end{align*}
By the hardness-vs-randomness results of Impagliazzo and Wigderson \cite{conf/ImpagliazzoW97}, Klivans and van Melkebeek \cite{journals/KlivansM02}, and Shaltiel and Umans \cite{journals/ShaltielU06}, this assumption gives logarithmic-seed pseudorandom generators against polynomial-size nonadaptive \(\SAT\)-oracle circuits.

Let \(\mathcal{R}\) be a randomized reduction with one-sided error from a source promise problem \(\Pi'\) to a target promise problem \(\Pi\).
For a fixed input \(x\), consider the success predicate \(\varphi_x(r)=1\) iff \(\mathcal{R}(x;r)\in \Pi_{\mathrm{YES}}\).
If this predicate is computable by a small nonadaptive \(\SAT\)-oracle circuit, then the pseudorandom generator produces a polynomial-size list of candidate random strings containing a successful one for every YES input.
On NO inputs, one-sidedness guarantees that every candidate string still produces a NO instance.
It remains to combine the polynomially many candidate outputs into a single many-one output.

This is done using OR functions, following Chang and Kadin \cite{ChangKadin95}.
For promise problems, an OR function maps a tuple of instances to a single instance that is YES if at least one input is YES, and is NO if all inputs are NO.
For parameterized promise problems, the output parameter must also be bounded by a function of the largest input parameter.
This formulation for promise problems is useful here because unsuccessful random tapes on YES inputs may produce outputs outside the promise.

The OR functions for MDP and SVP are block-diagonal constructions.
For MDP, given instances \((\mathbf{G}_i,k_i)\), let \(\kappa=\max_i k_i\) and \(K=\operatorname{lcm}(1,2,\ldots,\kappa)\).
After repeating every row of \(\mathbf{G}_i\) exactly \(K/k_i\) times, we output
\begin{align*}
    \bigl(\operatorname{diag}(\widetilde{\mathbf{G}}_1,\ldots,\widetilde{\mathbf{G}}_t), K\bigr),
\end{align*}
where \(\widetilde{\mathbf{G}}_i\) denotes the repeated matrix.
The minimum distance of the output code is the minimum of the scaled minimum distances of the blocks.

For SVP, given instances \((\mathbf{B}_i,d_i)\), we similarly scale the bases so that all thresholds become a common value \(D\).
Writing \(c_i d_i=D\), we output
\begin{align*}
    \bigl(\operatorname{diag}(c_1\mathbf{B}_1,\ldots,c_t\mathbf{B}_t), D\bigr).
\end{align*}
The shortest-vector length of a block-diagonal lattice is the minimum of the shortest-vector lengths of the scaled blocks, so this is an OR function.
In the parameterized setting, we take \(D=\operatorname{lcm}(1,2,\ldots,\kappa)\), where \(\kappa=\max_i d_i\), so the output parameter depends only on \(\kappa\).

Combining these OR functions with the pseudorandom-generator argument gives the conditional deterministic hardness results for MDP and SVP.
The same argument also applies in the non-parameterized setting to Micciancio's randomized reduction with one-sided error for \(\GapSVP_2\) \cite{journals/toc/Micciancio12}, yielding conditional deterministic \(\NP\)-hardness of \(\gamma\)-\(\GapSVP_2\) for every constant \(\gamma\ge 1\).
\subsection{Related Work}

Fine-grained hardness of lattice problems has been studied for Bounded Distance Decoding (BDD), SVP, and related problems under assumptions such as ETH, SETH, Gap-ETH, and Gap-SETH.
Aggarwal and Stephens-Davidowitz \cite{conf/stoc/AggarwalS18} proved fine-grained hardness results for SVP under Gap-ETH and SETH.
Li, Lin, and Liu \cite{conf/icalp/LiLL24} proved improved lower bounds under ETH and randomized ETH for approximating parameterized NCP and related problems.
Using earlier gap-preserving reductions \cite{journals/BhattacharyyaBE+21,conf/stoc/BennettCGR23}, they also obtained corresponding bounds for parameterized MDP, CVP, and SVP.
Bennett and Peikert \cite{conf/ccc/BennettP20} proved \(\NP\)-hardness and fine-grained hardness for BDD in \(\ell_p\) norms using locally dense lattices.
Bennett, Peikert, and Tang \cite{conf/itcs/BennettPT22} later obtained improved fine-grained hardness results for BDD and SVP under Gap-\((\mathrm{S})\)ETH.
Aggarwal, Gupta, Morolia, and Zhang \cite{aggarwal2026mindgapsvphardness} recently obtained hardness results under ETH for CVP, SVP, and BDD.
These works use locally dense lattice gadgets and sparsification, although the precise definitions and constructions differ from those used in this paper.

More recently, Hittmeir \cite{hittmeir2026finegrained} proved deterministic fine-grained hardness results for \(\GapSVP_p\) in which the dependence on the norm parameter \(p\) is part of the complexity statement.
In particular, Hittmeir considered a parameterized version of \(\GapSVP_p\) where \(p\) itself is part of the input and also serves as the parameter.
Under ETH, this yields lower bounds for algorithms solving \((2-\varepsilon)\)-\(\GapSVP_p\) uniformly over all \(p\in\mathbb{N}\).
Hittmeir also observed that this reduction has consequences for fixed-parameter tractability with respect to \(p\): \((2-\varepsilon)\)-\(\GapSVP_p\), parameterized by \(p\), does not belong to \(\mathsf{EPT}\), that is, it cannot be solved in time \(2^{O(p)}\cdot |x|^{O(1)}\), unless \(\P=\NP\).
This is different from the parameterized setting considered in this paper, where \(p\) is fixed and the distance threshold is the parameter.

OR functions and related combining reductions have also appeared in other contexts.
Chang and Kadin~\cite{ChangKadin95} studied polynomial-time OR functions for languages in their work on computing Boolean connectives of characteristic functions.
More recently, Ball, Boyle, Degwekar, Deshpande, Rosen, Vaikuntanathan, and Vasudevan~\cite{conf/itcs/BallBDDRVV20} studied lossy reductions that combine the values of several input instances according to a Boolean function; although their notion and use of OR reductions differ from ours, they showed that sufficiently lossy OR reductions can have cryptographic consequences, such as deriving one-way functions from worst-case hardness.

There is also an earlier conditional route to deterministic hardness for SVP.
Micciancio~\cite{journals/sicomp/Micciancio01} proved a \(\NP\)-hardness result for approximating SVP, namely hardness under deterministic many-one reductions, under a reasonable number-theoretic conjecture on the distribution of square-free smooth numbers.
\subsection{Open Problems}
Our results leave several natural open questions.
\begin{itemize}
    \item
    Can the circuit lower-bound assumption in \Cref{thm:MDP-is-W[1]hard-conditional,thm:SVP-is-W[1]hard-conditional} be removed?
    Equivalently, can one prove the \(\W[1]\)-hardness of \(\gamma\)-\(\GapMDP_q\) and \(\gamma\)-\(\GapSVP_p\) under deterministic \(\FPT\) many-one reductions unconditionally?

    \item
    For SVP, \Cref{thm:SVP-is-W[1]hard} gives randomized $\W[1]$-hardness with one-sided error for \(\gamma<2^{1/p}\).
    Bennett, Cheraghchi, Guruswami, and Ribeiro \cite{conf/stoc/BennettCGR23} obtained randomized hardness with two-sided error for a larger range, namely every constant factor for \(p>1\).
    It remains open whether their entire approximation range can be made one-sided error.

    \item
    For the Euclidean norm, Wan \cite{wan2026} already proved deterministic polynomial-time \(\NP\)-hardness of \(\gamma\)-\(\GapSVP_2\) for every constant \(1\le \gamma < \sqrt{2}\).
    It remains open whether deterministic polynomial-time many-one hardness can be extended to every constant \(\gamma \ge \sqrt{2}\).
\end{itemize}

\subsection{Organization}

\Cref{sec:prelim} contains the definitions and preliminary tools used throughout the paper.
\Cref{sec:hardness} proves the \(\W[1]\)-hardness results for MDP and SVP under randomized reductions with one-sided error.
\Cref{sec:derandomize} proves the general framework for converting randomized reductions with one-sided error into deterministic reductions under the circuit lower-bound assumption.
\Cref{sec:conditional} verifies the required OR functions and success-predicate conditions for MDP and SVP, and derives the conditional deterministic hardness theorems.

\section{Preliminaries}
\label{sec:prelim}
\subsection{Notations}
All vectors are column vectors unless stated otherwise.
For a positive integer $N$, we write $[N]:=\{1,\ldots,N\}$.
We write $\mathbb{Z}_{>0}$ for the set of positive integers.
For a finite set $S$, its cardinality is denoted by $|S|$; for a string $x$, its bit length is denoted by $|x|$.
We use $\langle x_1,\ldots,x_t\rangle$ for a standard polynomial-time computable encoding of a tuple.
Unless explicitly indicated otherwise, logarithms are base $2$; $\log_q$ denotes logarithm with base $q$.
We use $\mathrm{poly}(\cdot)$ for an unspecified polynomial and allow constants hidden in $O(\cdot)$ to depend on fixed constants such as $p$ and $q$, and on $\gamma$ whenever $\gamma$ is fixed throughout the statement.
We write $U_n$ for the uniform distribution over $\{0,1\}^n$.

For a map $f \colon X \rightarrow Y$ and subset $S \subseteq X$, we write $f(S):=\{f(x):x\in S\}\subseteq Y$ for the image of $S$ under $f$.
For a set $S$ in an abelian group and an element $a$, we write $S-a:=\{s-a:s\in S\}$.

For $\vec{x}=(x_1,\ldots,x_m)$, its support is $\operatorname{supp}(\vec{x}):=\{i\in[m]:x_i\neq 0\}$.
For vectors over either $\mathbb{F}_q$ or $\mathbb{Z}$, we write $\|\vec{x}\|_0:=|\operatorname{supp}(\vec{x})|$.
For $\vec{x}\in\mathbb{R}^m$ and $p\in[1,\infty)$, we write $\|\vec{x}\|_p:=(\sum_{i=1}^m |x_i|^p)^{1/p}$, and $\|\vec{x}\|_\infty:=\max_i |x_i|$.
The Hamming distance between two vectors over $\mathbb{F}_q$ is $\|\vec{x}-\vec{y}\|_0$.
For $r\ge 0$, define the Hamming ball over $\mathbb{F}_q^m$ by $B_{q,m}(r):=\{\vec{x}\in\mathbb{F}_q^m:\|\vec{x}\|_0\le r\}$ and the $\ell_p$ ball in $\mathbb{R}^m$ by $\mathcal{B}_p^m(r):=\{\vec{x}\in\mathbb{R}^m:\|\vec{x}\|_p\le r\}$.
When integer vectors are intended, we use $\mathcal{B}_p^m(r)\cap\mathbb{Z}^m$ implicitly.

\subsection{Probability Theory}
We use a standard consequence of Chebyshev's inequality.
\begin{lemma}
\label{lem:chebyshev}
Let $X_1,\ldots,X_N$ be pairwise independent random variables over $\{0, 1\}$ such that $\Pr[X_i = 1] = p > 0$ for $i = 1, \ldots, N$.
Then, it holds that
\begin{align*}
    \Pr[\forall i \in [N], X_i = 0] \leq \frac{1}{pN}.
\end{align*}
\end{lemma}

\subsection{Promise Problems}
We recall the basic definitions of promise problems and parameterized promise problems in complexity theory.
A promise problem is formalized by the following definition.
For details on promise problems, see the survey by Goldreich \cite{survey/Goldreich06}.
\begin{definition}[Promise problems]
A \emph{promise problem} is a pair of disjoint languages $\Pi=(\Pi_{\mathrm{YES}},\Pi_{\mathrm{NO}})$.
The set $\Pi_{\mathrm{YES}} \cup \Pi_{\mathrm{NO}}$ is called the \emph{promise}.
\end{definition}

We use the following parameterized version of promise problems.
For details on parameterized complexity, see \cite{book/DowneyF99,journals/tcs/DF95b,journals/sicomp/DF95a}.

\begin{definition}[Parameterized promise problems]
A \emph{parameterized promise problem} is a pair of disjoint languages $\Pi=(\Pi_{\mathrm{YES}},\Pi_{\mathrm{NO}})$ over $\Sigma^* \times \mathbb{N}$.
An instance of $\Pi$ is a pair $(x,k)\in \Sigma^* \times \mathbb{N}$, where $k$ is called the \emph{parameter}.
The set $\Pi_{\mathrm{YES}} \cup \Pi_{\mathrm{NO}}$ is called the \emph{promise}.
\end{definition}

Some reductions used are randomized.
They are required to preserve NO instances for every choice of randomness, while YES instances need only be mapped correctly with constant probability.
\begin{definition}[Randomized many-one reductions with one-sided error]
\label{def:one-sided-red}
Let $\Pi=(\Pi_{\mathrm{YES}},\Pi_{\mathrm{NO}})$ and $\Pi'=(\Pi'_{\mathrm{YES}},\Pi'_{\mathrm{NO}})$ be promise problems.
A randomized polynomial-time algorithm $R$ is called a randomized many-one reduction with one-sided error from $\Pi$ to $\Pi'$ if, on every input $x$, it outputs a single instance $R(x;r)$ in the input format of $\Pi'$, and the following hold:
\begin{enumerate}
    \item if $x\in \Pi_{\mathrm{YES}}$, then $\Pr_r[R(x;r)\in \Pi'_{\mathrm{YES}}]\ge 2/3$;
    \item if $x\in \Pi_{\mathrm{NO}}$, then $\Pr_r[R(x;r)\in \Pi'_{\mathrm{NO}}]=1$.
\end{enumerate}
In particular, for YES inputs we allow unsuccessful random strings to produce outputs outside the promise of $\Pi'$.
\end{definition}

For parameterized problems, we also require that the running time be fixed-parameter tractable and that the output parameter depend only on the input parameter.
This gives the parameterized analogue of the preceding randomized reduction.

\begin{definition}[Randomized $\FPT$ many-one reductions with one-sided error]
\label{def:one-sided-fpt-red}
Let $\Pi=(\Pi_{\mathrm{YES}},\Pi_{\mathrm{NO}})$ and $\Pi'=(\Pi'_{\mathrm{YES}},\Pi'_{\mathrm{NO}})$ be parameterized promise problems.
A randomized algorithm $R$ is called a randomized $\FPT$ many-one reduction with one-sided error from $\Pi$ to $\Pi'$ if there exist a computable function $T$, a computable function $g$, and a constant $c>0$ such that, on input $(x,k)$, the algorithm runs in time at most $T(k)\cdot |x|^c$, outputs a pair $(x',k')$, and satisfies
\begin{enumerate}
    \item $k'\le g(k)$;
    \item if $(x,k)\in \Pi_{\mathrm{YES}}$, then $\Pr[(x',k')\in \Pi'_{\mathrm{YES}}]\ge 2/3$;
    \item if $(x,k)\in \Pi_{\mathrm{NO}}$, then $\Pr[(x',k')\in \Pi'_{\mathrm{NO}}]=1$.
\end{enumerate}
\end{definition}

\begin{remark}\label{rem:fpt-time-encoding}
With respect to the fixed encoding of tuples, the running-time bounds
\begin{align*}
    T(k)\cdot |x|^c
    \qquad\text{and}\qquad
    T'(k)\cdot |\langle x,k\rangle|^c
\end{align*}
are equivalent up to replacing \(T\) by another computable function of \(k\).
Therefore, in \Cref{def:one-sided-fpt-red} (and deterministic $\FPT$ many-one reduction), one may measure the polynomial factor using either \(|x|\) or \(|\langle x,k\rangle|\).
\end{remark}

We assume, without loss of generality, that on each input the reduction uses a fixed number of random bits depending only on the input, and that this number is computable within the running time bound.
Variable-length random tapes can be padded to a fixed polynomial upper bound.
\subsection{Coding Problems}
For a matrix $\mathbf{G}\in\mathbb{F}_q^{m\times n}$, the linear code generated by the columns of $\mathbf{G}$ is $\mathcal{C}=\mathcal{C}(\mathbf{G}):=\{\mathbf{G}\vec{x}:\vec{x}\in\mathbb{F}_q^n\}\subseteq\mathbb{F}_q^m$.
The minimum distance of a nonzero linear code $\mathcal{C}\subseteq\mathbb{F}_q^m$ is $\lambda(\mathcal{C}):=\min_{\vec{c}\in\mathcal{C}\setminus\{\mathbf{0}\}}\|\vec{c}\|_0$.
We use the convention $\lambda(\{\mathbf{0}\}) := +\infty$.
For a target vector $\vec{t}\in\mathbb{F}_q^m$, we write $\operatorname{dist}(\mathcal{C},\vec{t}):=\min_{\mathbf{c}\in\mathcal{C}}\|\vec{c}-\vec{t}\|_0$.

The Minimum Distance Problem (MDP) asks for the Hamming weight of the shortest nonzero codeword in a given linear code.
In this paper we use its gap version, where one distinguishes codes of minimum distance at most $k$ from those of minimum distance greater than $\gamma k$.

\begin{definition}[$\gamma$-$\GapMDP_q$, \cite{journals/DumerMS03}]
\label{def:mdp}
For $\gamma \geq 1$ and a prime power $q \geq 2$, the (parameterized) promise problem $\gamma$-approximate Minimum Distance Problem over $\mathbb{F}_q$ ($\gamma$-$\GapMDP_q$) is defined as follows.
Instances are pairs $(\mathbf{G}, k)$, where $\mathbf{G} \in \mathbb{F}^{m\times n}_q$ is a generator matrix of $\mathcal{C}(\mathbf{G})$, and $k \in \mathbb{Z}_{>0}$ is a distance parameter such that
\begin{enumerate}
    \item $(\mathbf{G}, k)$ is a YES instance if $\lambda(\mathcal{C}(\mathbf{G}))\leq k$;
    \item $(\mathbf{G}, k)$ is a NO instance if $\lambda(\mathcal{C}(\mathbf{G}))> \gamma k$.
\end{enumerate}
The parameter of interest is $k$.
\end{definition}

The Nearest Codeword Problem (NCP) is the inhomogeneous analogue of MDP: given a target vector, it asks whether some codeword is close to the target.
We will use NCP as the starting point for reductions to MDP.

\begin{definition}[$\gamma$-$\GapNCP_q$, \cite{journals/DumerMS03}]
\label{def:ncp}
For $\gamma\geq 1$ and a prime power $q \geq 2$, the (parameterized) promise problem $\gamma$-approximate Nearest Codeword Problem over $\mathbb{F}_q$ ($\gamma$-$\GapNCP_q$) is defined as follows.
Instances are tuples $(\mathbf{G}, \vec{t}, k)$, where $\mathbf{G} \in \mathbb{F}^{m\times n}_q$ is a generator matrix of $\mathcal{C}(\mathbf{G})$, $\vec{t}\in \mathbb{F}_q^m $ is a target vector, and $k \in \mathbb{Z}_{>0}$ is a distance parameter such that
\begin{enumerate}
    \item $(\mathbf{G}, \vec{t}, k)$ is a YES instance if $\operatorname{dist}(\mathcal{C}(\mathbf{G}), \vec{t}) \leq k$;
    \item $(\mathbf{G}, \vec{t}, k)$ is a NO instance if $\operatorname{dist}(\mathcal{C}(\mathbf{G}), \vec{t}) > \gamma k$.
\end{enumerate}
The parameter of interest is $k$.
\end{definition}

The next theorem states the $\W[1]$-hardness of $\gamma$-$\GapNCP_q$.
We will use this result later in the reduction to $\gamma$-$\GapMDP_q$.
\begin{theorem}[{\cite[Theorem 5.1, adapted]{journals/BhattacharyyaBE+21} and \cite[Theorem 2.7]{conf/stoc/BennettCGR23}}]
\label{thm:NCP-is-hard}
For any real number $\gamma \geq 1$ and any prime power $q \geq 2$, $\gamma$-$\GapNCP_q$ is $\W[1]$-hard.
\end{theorem}

\subsection{Lattice Problems}
We now turn to the lattice problems considered in the paper.
The lattice analogues of the above coding problems are measured with respect to $\ell_p$ norms.

For a matrix $\mathbf{B}\in\mathbb{Z}^{m\times n}$, the lattice generated by the columns of $\mathbf{B}$ is $\mathcal{L}=\mathcal{L}(\mathbf{B}):=\{\mathbf{B}\vec{z}:\vec{z}\in\mathbb{Z}^n\}\subseteq\mathbb{Z}^m$; that is, $\mathcal{L}(\mathbf{B})$ consists of all integer linear combinations of the columns of $\mathbf{B}$.
If the columns of $\mathbf{B}$ are linearly independent, then $\mathbf{B}$ is called a lattice basis of $\mathcal{L}(\mathbf{B})$.
For $p\in[1,\infty]$, the $\ell_p$ length of a shortest nonzero lattice vector is  $\lambda_1^{(p)}(\mathcal{L}):=\min_{\vec{v}\in\mathcal{L}\setminus\{\mathbf{0}\}}\|\vec{v}\|_p$.
For a target vector $\vec{t}\in\mathbb{R}^m$, we write  $\operatorname{dist}_p(\mathcal{L},\vec{t}):=\min_{\vec{v}\in\mathcal{L}}\|\vec{v}-\vec{t}\|_p$.

One of the central problems on lattices is the Shortest Vector Problem (SVP), which is to find a shortest non-zero vector in a given lattice.
In this paper, we study the promise problem version of the Shortest Vector Problem, which is defined as follows:
\begin{definition}[$\gamma$-$\GapSVP_p$, \cite{journals/GoldreichG00}]
For $p \geq 1$ and $\gamma \geq 1$, the (parameterized) promise problem $\gamma$-approximate Shortest Vector Problem ($\gamma$-$\GapSVP_p$) is defined as follows.
Instances are pairs $(\mathbf{B}, d)$, where $\mathbf{B} \in \mathbb{Z}^{m\times n}$ is a lattice basis of $\mathcal{L}(\mathbf{B})$, and $d \in \mathbb{Z}_{>0}$ is a distance parameter such that
\begin{enumerate}
    \item $(\mathbf{B}, d)$ is a YES instance if $\lambda^{(p)}_1(\mathcal{L}(\mathbf{B}))\leq d$;
    \item $(\mathbf{B}, d)$ is a NO instance if $\lambda^{(p)}_1(\mathcal{L}(\mathbf{B}))> \gamma d$.
\end{enumerate}
The parameter of interest is $d$.
\end{definition}

We will also use the following classical $\NP$-hardness result for approximate SVP for $\ell_2$ norm.
This theorem is the non-parameterized hardness input for the deterministic hardness result later in the paper.

\begin{theorem}[{\cite[Corollary 6.3]{journals/toc/Micciancio12}}]
\label{thm:SVP-is-NPhard}
For any constant $\gamma \geq 1$, $\gamma$-$\GapSVP_2$ is $\NP$-hard under randomized many-one reductions with one-sided error.
\end{theorem}

The Closest Vector Problem (CVP) is the inhomogeneous analogue of SVP: given a lattice and a target vector, it asks whether the target is close to some lattice vector.
For the reduction to GapSVP, we use the following gap variant, whose NO instances require all lattice vectors to be far from every nonzero integer multiple of the target vector.
\begin{definition}[$\gamma$-$\GapCVP_p$, \cite{conf/stoc/BennettCGR23}]
\label{def:cvp}
For $p \geq 1$ and $\gamma\geq 1$, the (parameterized) promise problem $\gamma$-approximate Closest Vector Problem ($\gamma$-$\GapCVP_p$) is defined as follows.
Instances are tuples $(\mathbf{B}, \vec{t}, d)$, where $\mathbf{B} \in \mathbb{Z}^{m\times n}$ is a lattice basis of $\mathcal{L}(\mathbf{B})$, $\vec{t} \in \mathbb{Z}^m $ is a target vector, and $d \in \mathbb{Z}_{>0}$ is a distance parameter such that
\begin{enumerate}
    \item $(\mathbf{B}, \vec{t}, d)$ is a YES instance if $\operatorname{dist}_p(\mathcal{L}(\mathbf{B}), \vec{t}) \leq d$;
    \item $(\mathbf{B}, \vec{t}, d)$ is a NO instance if $\operatorname{dist}_p(\mathcal{L}(\mathbf{B}), w\vec{t}) > \gamma d$ for all $w\in\mathbb{Z}\backslash\{0\}$.
\end{enumerate}
The parameter of interest is $d$.
\end{definition}

The following theorem provides the parameterized hardness of this CVP variant.
It will be combined with our randomized reduction from CVP to SVP.
\begin{theorem}[{\cite[Theorem 7.2]{journals/BhattacharyyaBE+21}} and {\cite[Theorem 2.11]{conf/stoc/BennettCGR23}}]
\label{thm:CVP-is-hard}
For any real numbers $\gamma \geq 1$ and $p \geq 1$, $\gamma$-$\GapCVP_p$ is $\W[1]$-hard under deterministic $\FPT$ many-one reductions.
\end{theorem}
\begin{remark}
In \cite{journals/BhattacharyyaBE+21}, the NO instance condition for \(\gamma\)-\(\GapCVP_p\) is stated as $\operatorname{dist}_p(\mathcal{L}(\mathbf{B}), \vec{t}) > \gamma d$, but their proof in fact shows that $\operatorname{dist}_p(\mathcal{L}(\mathbf{B}), w\vec{t}) > \gamma d$ for every nonzero integer \(w \in \mathbb{Z}\setminus\{0\}\).
\end{remark}
\subsection{OR Functions}
Following Chang and Kadin~\cite{ChangKadin95}, for a language $A\subseteq\{0,1\}^*$ and an integer $t\ge 1$, define
\begin{align*}
    \OR_t(A)
    &:=
    \{\langle x_1,\ldots,x_t\rangle : \exists i\in[t]\text{ such that }x_i\in A\},
    \\
    \OR_\omega(A)
    &:=
    \bigcup_{t\ge 1} \OR_t(A).
\end{align*}
Chang and Kadin say that $A$ has a polynomial-time OR function if $\OR_\omega(A) \le_m^p A$.
We use the following promise problem analogue, which is the right notion when randomized reductions may output garbage on unsuccessful random strings.
The definition only constrains tuples consisting of at least one YES instance or entirely of NO instances.

\begin{definition}[OR functions of a promise problem]\label{def:or-promise}
Let $\Pi=(\Pi_{\textrm{YES}},\Pi_{\textrm{NO}})$ be a promise problem.
We say that $\Pi$ has a polynomial-time OR function if there exists a deterministic polynomial-time algorithm $C$ such that, for every $t\ge 1$ and every tuple $(x_1,\ldots,x_t)$ of strings,
\begin{enumerate}
    \item if $x_i\in \Pi_{\textrm{YES}}$ for some $i\in[t]$, then $C(x_1,\ldots,x_t)\in \Pi_{\textrm{YES}}$;
    \item if $x_i\in \Pi_{\textrm{NO}}$ for all $i\in[t]$, then $C(x_1,\ldots,x_t)\in \Pi_{\textrm{NO}}$.
\end{enumerate}
No condition is imposed when no input belongs to $\Pi_{\textrm{YES}}$ and at least one input lies outside $\Pi_{\textrm{NO}}$.
\end{definition}

For parameterized promise problems, an OR function must also keep the parameter under control.
We require that the output parameter be bounded by a computable function of the largest input parameter.

\begin{definition}[FPT OR functions of a parameterized promise problem]\label{def:or-para}
Let $\Pi=(\Pi_{\textrm{YES}},\Pi_{\textrm{NO}})$ be a parameterized promise problem.
We say that $\Pi$ has a FPT OR function if there exist a deterministic algorithm $C$, computable functions $T$ and $g$, and a constant $c>0$ such that, on every tuple $\bigl((x_1,k_1),\ldots,(x_t,k_t)\bigr)$, the algorithm outputs an instance
\begin{align*}
    C\bigl((x_1,k_1),\ldots,(x_t,k_t)\bigr)=(y,\ell)
\end{align*}
in time at most
\begin{align*}
    T(\kappa)\cdot N^c,
    \qquad
    \kappa:=\max_{i\in[t]} k_i,
    \qquad
    N:=\sum_{i=1}^t |\langle x_i,k_i\rangle|.
\end{align*}
and satisfies:
\begin{enumerate}
    \item if $(x_i,k_i)\in \Pi_{\textrm{YES}}$ for some $i\in[t]$, then $(y,\ell)\in \Pi_{\textrm{YES}}$;
    \item if $(x_i,k_i)\in \Pi_{\textrm{NO}}$ for all $i\in[t]$, then $(y,\ell)\in \Pi_{\textrm{NO}}$;
    \item $\ell\le g(\kappa)$.
\end{enumerate}
No condition is imposed when there is no YES input and at least one input lies outside $\Pi_{\textrm{YES}}\cup \Pi_{\textrm{NO}}$.
\end{definition}

\subsection{Circuits and Pseudorandom Generators}
Finally, we recall the circuit classes and pseudorandom generators used in the derandomization argument.
These notions allow us to replace the random choices in one-sided reductions by a polynomial-size list of pseudorandom seeds.

A nondeterministic circuit is the nonuniform analogue of an $\NP$ verifier: it has ordinary input wires together with auxiliary nondeterministic wires, and it accepts when some auxiliary assignment makes the circuit output $1$.

\begin{definition}[Nondeterministic circuits, see also \cite{eccc/ShaltielS23,conf/ShaltielS24}]
A nondeterministic circuit is a Boolean circuit with additional nondeterministic input wires.
Such a circuit accepts an input $x$ if there exists an assignment to the nondeterministic input wires that makes the circuit output $1$ on $x$.
The size of a circuit is the total number of gates and wires.

For a language $L \subseteq \{0,1\}^*$ and an integer $n \in \mathbb{N}$, let $L_n$ denote the restriction of $L$ to inputs of length $n$.
For a function $s\colon\mathbb N \to \mathbb N$, $\mathrm{NSIZE}(s(n))$ denotes the class of languages decided by families of nondeterministic circuits of size $O(s(n))$.
We write $L \in i.o.\text{-}\mathrm{NSIZE}(s(n))$ if there exists a language $L' \in \mathrm{NSIZE}(s(n))$ such that $L_n = L'_n$ for infinitely many input lengths $n$.
\end{definition}

We also use Boolean circuits that are allowed to query a fixed Boolean function.
In the applications below, the fixed function is $\SAT$, and the nonadaptive restriction means that no input-to-output path in the circuit passes through more than one oracle gate.

\begin{definition}[Oracle circuits, see also \cite{eccc/ShaltielS23,conf/ShaltielS24}]
For a Boolean function $A$, an $A$-oracle circuit is a Boolean circuit that, in addition to the standard gates, may use oracle gates for $A$.
An $A$-oracle circuit is called \emph{nonadaptive} if on every path from an input gate to the output gate there is at most one oracle gate.
The size of a circuit is the total number of gates and wires.
For a function $s:\mathbb N\to\mathbb N$, $\mathrm{SIZE}^{A}_{\|}(s(n))$ denotes the class of languages decided by families of nonadaptive $A$-oracle circuits of size $O(s(n))$.
\end{definition}

To apply pseudorandom generators to a randomized $\FPT$ reduction, we need the predicate saying that a random string is successful to be recognizable by a small nonadaptive $\SAT$-oracle circuit.
The following definition isolates this property of reductions.

\begin{definition}[SAT-circuitizable success predicates]
\label{def:sat-circuitizable}
Let $\Pi=(\Pi_{\mathrm{YES}},\Pi_{\mathrm{NO}})$ and $\Pi'=(\Pi'_{\mathrm{YES}},\Pi'_{\mathrm{NO}})$ be parameterized promise problems, and let $R$ be a randomized $\FPT$ many-one reduction with one-sided error from $\Pi$ to $\Pi'$.
We say that $R$ has SAT-circuitizable success predicates if there exists a fixed polynomial $P$ such that for every input $(x,k)$, letting $n:=|\langle x,k\rangle|$, letting $\mu$ be the number of random bits used by $R$ on input $(x,k)$, and letting $m:=n+\mu$, the Boolean function
\begin{align*}
    \phi_{x,k}(r)=1 \iff R((x,k);r)\in \Pi'_{\mathrm{YES}}
\end{align*}
is computable by a nonadaptive $\SAT$-oracle circuit of size at most $P(m)$.
\end{definition}

A pseudorandom generator stretches a short seed into a longer string that is indistinguishable from a uniform random string for a specified class of Boolean circuits.
We use the standard formulation in terms of distinguishing advantage.

\begin{definition}[Pseudorandom generators against circuit classes]\label{def:prg}
Let $\mathcal C=\{\mathcal C_n\}_{n\ge 1}$ be a family of Boolean circuit classes, where each $\mathcal C_n$ consists of Boolean functions on $n$ input bits.
For functions $m \colon \mathbb N\to\mathbb N$ and a function $\varepsilon \colon \mathbb N\to [0,1]$, a family of functions
\begin{align*}
    G_n:\{0,1\}^{m(n)}\to\{0,1\}^{n}
\end{align*}
is called an $\varepsilon(n)$-PRG for $\mathcal C$ if for every sufficiently large $n$ and every Boolean function $D\in \mathcal C_n$,
\begin{align*}
\left| \Pr_{y\leftarrow U_n}[D(y)=1] - \Pr_{\sigma\leftarrow U_{m(n)}}[D(G_n(\sigma))=1] \right| \le \varepsilon(n),
\end{align*}
where $U_t$ denotes the uniform distribution on $\{0,1\}^t$.
\end{definition}

The next theorem is the derandomization tool used in the paper.
Under an exponential circuit lower bound for $\E := \DTIME(2^{O(n)})$ against nondeterministic circuits, it gives logarithmic-seed PRGs that fool polynomial-size nonadaptive $\SAT$-oracle circuits.

\begin{theorem}[\cite{conf/ImpagliazzoW97,journals/KlivansM02,journals/ShaltielU06}, see also Theorem 2.8 in \cite{eccc/ShaltielS23}]
\label{thm:prg-under-assumption}
Assume that there exists an $\varepsilon > 0$ such that $\E \nsubseteq i.o.\text{-}\mathrm{NSIZE}(2^{\varepsilon n})$.
Then, for every constant $c > 1$, there exists a constant $a > 1$ such that, for all sufficiently large $n$, there is a $1/n^c$-PRG
\begin{align*}
G_n:\{0,1\}^{a\log n}\to\{0,1\}^n
\end{align*}
for $\mathrm{SIZE}^{\SAT}_{\|}(n^c)$, and $G_n$ is computable in time $\mathrm{poly}(n^c)$.
\end{theorem}
\section{W[1]-hardness with One-Sided Error}
\label{sec:hardness}
In this section, we will show that $\gamma$-$\GapMDP_q$ and $\gamma$-$\GapSVP_p$ are $\W[1]$-hard under randomized $\FPT$ reductions with one-sided error.
As in the prior work \cite{conf/stoc/BennettCGR23}, this reduction is based on locally dense codes/lattices constructed from BCH and Reed--Solomon codes, together with sparsification; however, by additionally incorporating a gadget based on the minimum distance of the locally dense codes/lattices, it achieves one-sided error.

Locally dense codes and locally dense lattices are defined as follows.
\begin{definition}[Locally dense codes, \cite{conf/stoc/BennettCGR23}]
Fix a real number $\alpha\in(0,1)$, positive integers $d, N, m, n$, and a prime power $q$.
A \emph{$(q,\alpha,d, N, m, n)$-locally dense code} is specified by a generator matrix $\mathbf{A} \in \F_q^{m \times n}$ and a target vector $\vec{s} \in \F_q^m$ with the following properties:
\begin{itemize}
    \item $\lambda(\C(\mathbf{A})) > d$.
    \item $\card{(\C(\mathbf{A}) - \vec{s}) \cap B_{q, m}(\alpha d)} \geq N$.
\end{itemize}
\end{definition}

\begin{definition}[Locally dense lattices, \cite{conf/stoc/BennettCGR23}]
Fix real numbers $\alpha\in(0,1)$ and $p\geq 1$ and positive integers $d, N, m, n$.
A \emph{$(p,\alpha,d, N, m, n)$-locally dense lattice} is specified by a basis $\mathbf{A} \in \Z^{m \times n}$ and a target vector $\vec{s} \in \Z^m$ with the following properties:
\begin{itemize}
    \item $\lambda^{(p)}_1(\lat(\mathbf{A})) > d$.
    \item $\card{(\lat(\mathbf{A}) - \vec{s}) \cap \mathcal{B}_{p}^{m}(\alpha d)} \geq N$.
\end{itemize}
\end{definition}

\subsection{W[1]-hardness of \texorpdfstring{$\gamma$-$\GapMDP_q$}{MDP} with One-Sided Error}
We will show the $\W[1]$-hardness of GapMDP by giving a randomized $\FPT$ reduction with one-sided error from GapNCP to GapMDP.

For this, we use properties of BCH codes.
\begin{theorem}[q-ary BCH codes, \cite{conf/stoc/BennettCGR23}]
\label{thm:bch-code}
Fix a prime power $q$.
Then, given integers $m' = q^r - 1$ and $1 \leq d \leq m'$ for some integer $r$, it is possible to construct in time $\rm{poly}(m')$ a generator matrix $\mathbf{G}_\mathrm{BCH} \in \mathbb{F}_q^{m' \times n'}$ such that $\mathcal{C}_\mathrm{BCH} = \mathcal{C}(\mathbf{G}_\mathrm{BCH}) \subseteq \mathbb{F}_q^{m'}$ has minimum distance at least $d$ and co-dimension
\begin{align*}
    m' - n' \leq \lceil(d-1)(1-1/q)\rceil \log_q (m'+1).
\end{align*}
\end{theorem}

Bennett, Cheraghchi, Guruswami, and Ribeiro \cite{conf/stoc/BennettCGR23} constructed the following locally dense code using the BCH code from \Cref{thm:bch-code}.
\begin{theorem}[{\cite[Theorem 3.3]{conf/stoc/BennettCGR23}}]
\label{thm:algo-ldc}
Fix a prime power $q \geq 2$ and set $\gamma = 4q$.
There exists a randomized algorithm which, given positive integers \(m\) and \(k\le m\), runs in time $\rm{poly}(m)$ and outputs with probability at least 0.99 a $(q,\alpha,d,N, m', n')$-locally dense code $(\mathbf{A}, \vec{s})$, where
\begin{align*}
    &m', n' \leq \mathrm{poly}(m), \\
    &d = \gamma k = 4qk, \\
    &\alpha = 1 - \frac{1}{2q}, \\
    &N = \frac{(qm)^{2d}}{100},
\end{align*}
provided that $m$ is sufficiently large compared to $q$.
Moreover, $\lambda(\mathcal{C}(\mathbf{A})) > d$ with probability 1.
\end{theorem}

In addition to these locally dense codes, by using a random matrix and the parity-check matrix of a BCH code, we obtain a randomized $\FPT$ reduction with one-sided error.
\begin{theorem}
\label{thm:NCP-to-MDP}
For a fixed prime power $q \ge 2$ and real numbers $\gamma := 4q$ and $\gamma' := \gamma/(\gamma - 1) = 4q/(4q-1)$, there exists a randomized $\FPT$ many-one reduction with one-sided error from $\gamma$-$\GapNCP_q$ to $\gamma'$-$\GapMDP_q$.
\end{theorem}
\begin{proof}
We use the following equivalent formulation of $4q$-$\GapNCP_q$.
An instance $(\mathbf G,\vec t,k)$ consists of a matrix $\mathbf G\in\mathbb F_q^{m\times n}$, a target $\vec t\in\mathbb F_q^m$, and an integer $k\ge 1$.
In the NO case we are promised that
\begin{align*}
    \operatorname{dist}(\mathcal C(\mathbf G),\beta \vec t)>4qk
    \qquad\text{for every }\beta\in\mathbb F_q\setminus \{0\}.
\end{align*}
This is equivalent to the usual definition because, for every $\beta\in\mathbb F_q\setminus \{0\}$,
\begin{align*}
    \operatorname{dist}(\mathcal C(\mathbf G),\beta \vec t)
    =
    \min_{\vec c\in\mathcal C(\mathbf G)} \|\vec c-\beta \vec t\|_0
    =
    \min_{\vec c'\in\mathcal C(\mathbf G)} \|\beta \vec c' - \beta \vec t\|_0
    =
    \min_{\vec c'\in\mathcal C(\mathbf G)} \|\vec c'-\vec t\|_0
    =
    \operatorname{dist}(\mathcal C(\mathbf G),\vec t),
\end{align*}
since $\mathcal C(\mathbf G)$ is closed under multiplication by nonzero field elements and $\|\beta \vec v\|_0=\|\vec v\|_0$ for all $\vec v\in\mathbb F_q^m$.

Without loss of generality for the following reasons, we may assume that $1 \le k \le m$ and $m$ is sufficiently large compared to $q$: Let $m_{\mathrm{ldc}}(q)$ be a constant such that \Cref{thm:algo-ldc} applies for all $m \ge m_{\mathrm{ldc}}(q)$.
Define
\begin{align*}
    M(q) := \min \left\{ M \in \mathbb{Z}_{\ge 1} : M \ge m_{\mathrm{ldc}}(q),\;
    (qM)^{8q} \ge 200,
    \text{ and }
    200q \left( \frac{4q+3}{qM} \right)^{4q} \le 0.01 \right\}.
\end{align*}
Since $q$ is fixed, $M(q)$ is an absolute constant.
On input $(\mathbf{G},\vec{t},k)$, we first perform the following deterministic preprocessing.
If $k > m$, then $(\mathbf{G},\vec{t},k)$ is automatically a YES instance, because every vector in $\mathbb{F}_q^m$ has $\ell_0$-norm at most $m < k$.
In this case we output any fixed YES instance of $\gamma'$-\GapMDP$_q$.
If $m < M(q)$, then we decide the instance exactly by exhaustive search over all vectors $\vec{e} \in \mathbb{F}_q^m$.
Indeed, $(\mathbf{G},\vec{t},k)$ is YES if and only if there exists $\vec{e} \in \mathbb{F}_q^m$ with $\|\vec{e}\|_0 \le k$ and $\vec{t}-\vec{e} \in \mathcal{C}(\mathbf{G})$.
Since $q$ and $M(q)$ are constants, this takes polynomial time.
We then output a fixed YES instance or a fixed NO instance of $\gamma'$-\GapMDP$_q$ accordingly.
Hence, in the sequel we may assume $1 \le k \le m$ and $m \ge M(q)$.

We first describe the reduction.
Set $d := \gamma k$, $\alpha := 1 - 1/(2q)$, and $k' := k + \alpha d = (4q-1)k$.
By \Cref{thm:algo-ldc}, in randomized polynomial time on input $(m,k)$ we obtain, with probability at least $0.99$, a $(q,\alpha,d,N,m',n')$-locally dense code $(\mathbf{A},\vec{s})$, where $m',n' \le \textrm{poly}(m)$ and $N=(qm)^{2d}/100$.
Moreover, \Cref{thm:algo-ldc} guarantees that $\lambda(\mathcal{C}(\mathbf{A})) > d$ with probability $1$.
Define
\begin{align*}
    \mathcal{S} := (\mathcal{C}(\mathbf{A})-\vec{s}) \cap B_{q,m'}(\alpha d).
\end{align*}
Then, with probability at least $0.99$, we have $|\mathcal{S}| \geq (qm)^{2d}/100$.
Next define
\begin{align*}
    \widehat{m} := \min \left\{ q^r - 1 : q^r - 1 \ge \max\{m,d+1\} \right\}.
\end{align*}
By \Cref{thm:bch-code} applied with block length $\widehat{m}$ and distance parameter $d+1$, there is a $q$-ary BCH code of length $\widehat{m}$ and minimum distance at least $d+1$ whose co-dimension $h$ satisfies
\begin{align*}
h \le \left\lceil d (1-1/q) \right\rceil \log_q(\widehat{m}+1) = d (1-1/q) \log_q(\widehat{m}+1).
\end{align*}
The equality holds because $d (1-1/q) = 4(q-1)k$ is an integer.
Let $\mathbf{H}_P \in \mathbb{F}_q^{h \times \widehat{m}}$ be a parity-check matrix of this BCH code.
Let $\mathbf{P} \in \mathbb{F}_q^{h \times m}$ be the matrix consisting of the first $m$ columns of $\mathbf{H}_P$.
Then every nonzero vector $\vec{u} \in \mathbb{F}_q^m$ with $\|\vec{u}\|_0 \le d$ satisfies $\mathbf{P}\vec{u} \neq 0$.
Indeed, if $\mathbf{P}\vec{u}=0$, then the vector obtained from $\vec{u}$ by padding $\widehat{m}-m$ zeros would be a nonzero codeword of the BCH code checked by $\mathbf{H}_P$ of $\ell_0$-norm at most $d$, contradicting its minimum distance at least $d+1$.
Define the linear map
\begin{align*}
    \mathbf{E} : \mathbb{F}_q^h \to \mathbb{F}_q^{(d+1)h}
\end{align*}
by
\begin{align*}
    \mathbf{E}(a_1,\dots,a_h) := (a_1 \vec{1}_{d+1}, \dots, a_h \vec{1}_{d+1}),
\end{align*}
where $\vec{1}_{d+1}$ is the all-ones vector of length $d+1$.
Then every nonzero $\vec{a} \in \mathbb{F}_q^h$ satisfies $\|\mathbf{E}\vec{a}\|_0 \ge d+1$.
Finally, sample $\mathbf{R} \in \mathbb{F}_q^{h \times m'}$ uniformly at random, and output the code generated by
\begin{align*}
    \mathbf{G}'
    :=
    \begin{pmatrix}
    \mathbf{G} & \mathbf{0} & -\vec{t} \\
    \mathbf{0} & \mathbf{A} & -\vec{s} \\
    \mathbf{E}\mathbf{P}\mathbf{G} & \mathbf{E}\mathbf{R}\mathbf{A} & -\mathbf{E}(\mathbf{P}\vec{t}+\mathbf{R}\vec{s})
    \end{pmatrix}
\end{align*}
together with the parameter $k'$.
Equivalently, for every $(\vec{x},\vec{y},\beta) \in \mathbb{F}_q^{n+n'+1}$, the corresponding codeword is $\vec{c}(\vec{x},\vec{y},\beta)=\bigl(\vec{u},\vec{w},\mathbf{E}(\mathbf{P}\vec{u}+\mathbf{R}\vec{w})\bigr)$, where $\vec{u} := \mathbf{G}\vec{x}-\beta \vec{t}$ and $\vec{w} := \mathbf{A}\vec{y}-\beta \vec{s}$.

We now verify that this is an $\FPT$ reduction.
By \Cref{thm:algo-ldc}, we have $m',n' \le \textrm{poly}(m)$.
Also $\widehat{m} \le q(\max\{m,d+1\}+1)=O_q(m)$ because $k \le m$ and hence $d=4qk \le 4qm$.
Therefore $h = O_q(k \log m)$, the output length is polynomial in the input size, and the output parameter is $k' = (4q-1)k$.
Hence the reduction runs in time $f(k)\cdot \textrm{poly}(|\mathbf{G}|+|\vec{t}|)$ and maps $k$ to $O(k)$.

It remains to prove completeness and soundness.

\paragraph*{Completeness.}
Assume that $(\mathbf{G},\vec{t},k)$ is a YES instance of $4q$-\GapNCP$_q$.
Then there exists $\vec{x}' \in \mathbb{F}_q^n$ such that $\|\mathbf{G}\vec{x}'-\vec{t}\|_0 \le k$.
Let $\vec{u}' := \mathbf{G}\vec{x}' - \vec{t}$ and $\vec{a} := -\mathbf{P}\vec{u}' \in \mathbb{F}_q^h$.
Condition on the event that $(\mathbf{A},\vec{s})$ is locally dense, equivalently that $|\mathcal{S}| \geq (qm)^{2d}/100$.
This event has probability at least $0.99$.

Choose a subset $\mathcal{T} \subseteq \mathcal{S}\setminus\{\vec{0}\}$ containing exactly one representative from each one-dimensional subspace that intersects $\mathcal{S}\setminus\{\vec{0}\}$ nontrivially.
Then distinct elements of $\mathcal{T}$ are pairwise linearly independent, and, since $(qm)^{2d}/100\ge 2$ by the definition of $M(q)$,
\begin{align*}
    |\mathcal{T}| \ge \frac{|\mathcal{S}|-1}{q-1} \ge \frac{|\mathcal{S}|}{2q} \geq \frac{(qm)^{2d}}{200q}.
\end{align*}
For each $\vec{w} \in \mathcal{T}$, define the indicator random variable $X_{\vec{w}} := \mathbf{1}[\mathbf{R}\vec{w}=\vec{a}]$.
Since $\vec{w} \neq 0$, for a uniformly random row $\vec{r} \in \mathbb{F}_q^{m'}$ the map $\vec{r} \mapsto \vec{r}\vec{w}$ is surjective onto $\mathbb{F}_q$, and hence $\Pr[X_{\vec{w}}=1] = q^{-h}$.
If $\vec{w},\vec{w}' \in \mathcal{T}$ are distinct, then they are linearly independent, so for a uniformly random row $\vec{r}$ the map $\vec{r} \mapsto (\vec{r}\vec{w},\vec{r}\vec{w}')$ is surjective onto $\mathbb{F}_q^2$.
Hence $\Pr[X_{\vec{w}}=1 \text{ and } X_{\vec{w}'}=1] = q^{-2h}$, and therefore the family $\{X_{\vec{w}}\}_{\vec{w}\in \mathcal{T}}$ is pairwise independent.

By \Cref{lem:chebyshev} applied with $p=q^{-h}$ and $N=|\mathcal{T}|$, we obtain
\begin{align*}
    \Pr[\forall \vec{w} \in \mathcal{T},\, X_{\vec{w}}=0] \le \frac{q^h}{|\mathcal{T}|}.
\end{align*}
Using the BCH codimension bound and the bound on $\widehat{m}$, we get
\begin{align*}
    q^h
    &\le (\widehat{m}+1)^{d(1-1/q)} \\
    &\le \bigl(q(\max\{m,d+1\}+1)\bigr)^d \\
    &\le \bigl(q((4q+1)m+2)\bigr)^d \\
    &\le \bigl(q(4q+3)m\bigr)^d.
\end{align*}
Therefore
\begin{align*}
    \Pr[\forall \vec{w} \in \mathcal{T},\, X_{\vec{w}}=0]
    &\le 200q \cdot \frac{q^h}{(qm)^{2d}} \\
    &\le 200q\left(\frac{q(4q+3)m}{q^2m^2}\right)^d \\
    &=200q\left(\frac{4q+3}{qm}\right)^d.
\end{align*}
Since $d=4qk \ge 4q$ and $m \ge M(q)$, the definition of $M(q)$ implies
\begin{align*}
    \Pr[\forall \vec{w} \in \mathcal{T},\, X_{\vec{w}}=0] \le 0.01.
\end{align*}
Hence, conditioned on the event that $(\mathbf{A},\vec{s})$ is locally dense, with probability at least $0.99$ over the choice of $\mathbf{R}$ there exists $\vec{w} \in \mathcal{T} \subseteq \mathcal{S}$ such that $\mathbf{R}\vec{w} = \vec{a} = -\mathbf{P}\vec{u}'$.
Fix such a $\vec{w}$.
Since $\vec{w} \in \mathcal{S} \subseteq \mathcal{C}(\mathbf{A})-\vec{s}$, there exists $\vec{y} \in \mathbb{F}_q^{n'}$ such that $\vec{w} = \mathbf{A}\vec{y} - \vec{s}$ and $\|\vec{w}\|_0 \le \alpha d$.
Then $\mathbf{P}\vec{u}' + \mathbf{R}\vec{w} = 0$, and therefore $\vec{c}(\vec{x}',\vec{y},1) = (\vec{u}',\vec{w},0)$.
It follows that
\begin{align*}
    \|\vec{c}(\vec{x}',\vec{y},1)\|_0 \le \|\vec{u}'\|_0 + \|\vec{w}\|_0 \le k + \alpha d = k'.
\end{align*}
Thus the output instance is a YES instance of $\gamma'$-\GapMDP$_q$.

Since the event that $(\mathbf{A},\vec{s})$ is locally dense has probability at least $0.99$ and, conditioned on that event, the probability over $\mathbf{R}$ that some $\vec{w} \in \mathcal{S}$ satisfies $\mathbf{R}\vec{w}=-\mathbf{P}\vec{u}'$ is at least $0.99$, we conclude that
\begin{align*}
    \Pr[\text{output is a YES instance of }\gamma'\text{-}\GapMDP_q] \ge 0.99 \cdot 0.99 > 2/3.
\end{align*}

\paragraph*{Soundness.}
Assume that $(\mathbf{G},\vec{t},k)$ is a NO instance of $4q$-\GapNCP$_q$.
We show that for every realization of $(\mathbf{A},\vec{s})$ produced above and every choice of $\mathbf{R}$, the output is a NO instance of $\gamma'$-\GapMDP$_q$.

Let $\vec{c} = \vec{c}(\vec{x},\vec{y},\beta) \in \mathcal{C}(\mathbf{G}')\setminus\{0\}$ be an arbitrary nonzero codeword.
Suppose toward a contradiction that $\|\vec{c}\|_0 \le d$.
Writing $\vec{c} = \bigl(\vec{u},\vec{w},\mathbf{E}(\mathbf{P}\vec{u}+\mathbf{R}\vec{w})\bigr)$, where $\vec{u} = \mathbf{G}\vec{x}-\beta \vec{t}$ and $\vec{w} = \mathbf{A}\vec{y}-\beta \vec{s}$, we get $\|\vec{u}\|_0 + \|\vec{w}\|_0 \le d$.

We first show that $\beta=0$.
If $\beta \neq 0$, then by the NO promise for $4q$-\GapNCP$_q$, $\|\vec{u}\|_0 = \|\mathbf{G}\vec{x}-\beta \vec{t}\|_0 > 4qk = d$, contradicting $\|\vec{u}\|_0 + \|\vec{w}\|_0 \le d$.

Thus $\beta=0$.
We next show that $\vec{w}=0$.
If $\vec{w}\neq 0$, then $\vec{w}=\mathbf{A}\vec{y}$ is a nonzero codeword of $\mathcal{C}(\mathbf{A})$, and hence $\|\vec{w}\|_0 \ge \lambda(\mathcal{C}(\mathbf{A})) > d$, again contradicting $\|\vec{u}\|_0 + \|\vec{w}\|_0 \le d$.

Therefore $\beta=0$ and $\vec{w}=0$, so $\vec{c} = \bigl(\vec{u},0,\mathbf{E}(\mathbf{P}\vec{u})\bigr)$ with $\vec{u} = \mathbf{G}\vec{x}$.
Since $\vec{c} \neq 0$, we must have $\vec{u} \neq 0$.
Also $\|\vec{u}\|_0 \le d$, so by the defining property of $\mathbf{P}$, $\mathbf{P}\vec{u} \neq 0$.
Hence $\|\mathbf{E}(\mathbf{P}\vec{u})\|_0 \ge d+1$, which implies $\|\vec{c}\|_0 \ge d+1$, contradicting $\|\vec{c}\|_0\le d$.
This contradiction proves that every nonzero codeword of $\mathcal{C}(\mathbf{G}')$ has $\ell_0$-norm strictly larger than $d$.
Therefore $\lambda(\mathcal{C}(\mathbf{G}')) > d$.

Finally, since $d = \gamma' k'$, we obtain $\lambda(\mathcal{C}(\mathbf{G}')) > \gamma' k'$, so the output is a NO instance of $\gamma'$-\GapMDP$_q$.
This holds for every realization of $(\mathbf{A},\vec{s})$ and every choice of $\mathbf{R}$, and thus the reduction has one-sided error.
\end{proof}

To amplify to an arbitrary approximation constant factor, we use the following well-known fact.
\begin{fact}
\label{fact:tensoring}
For any two linear codes $\mathcal{C}(\mathbf{G}_1)$ and $\mathcal{C}(\mathbf{G}_2)$ with $\mathbf{G}_1 \in \mathbb{F}_q^{m_1 \times n_1}$, $\mathbf{G}_2 \in \mathbb{F}_q^{m_2 \times n_2}$, and minimum distances $\lambda(\mathcal{C}(\mathbf{G}_1)) = d_1$, $\lambda(\mathcal{C}(\mathbf{G}_2)) = d_2$, we have
\begin{align*}
    \lambda(\mathcal{C}(\mathbf{G}_1 \otimes \mathbf{G}_2)) = d_1 \cdot d_2,
\end{align*}
where $\mathbf{G}_1 \otimes \mathbf{G}_2 \in \mathbb{F}_q^{m_1 m_2 \times n_1 n_2}$ is the Kronecker product of $\mathbf{G}_1$ and $\mathbf{G}_2$.
\end{fact}

By combining the above claims, we obtain the desired proof.
\begin{proof}[Proof of \Cref{thm:MDP-is-W[1]hard}]
By \Cref{thm:NCP-is-hard,thm:NCP-to-MDP}, the problem
\begin{align*}
    \gamma_0\text{-}\GapMDP_q
    \qquad\text{with}\qquad
    \gamma_0:=\frac{4q}{4q-1}>1
\end{align*}
is $\W[1]$-hard under randomized $\FPT$ many-one reductions with one-sided error.

Fix an integer $c\ge 1$.
Given an instance $(\mathbf G,k)$ of $\gamma_0$-$\GapMDP_q$, map it to $(\mathbf G^{\otimes c},k^c)$, where $\mathbf G^{\otimes c}$ denotes the $c$-fold Kronecker product of $\mathbf G$ with itself.
By \Cref{fact:tensoring}, $\lambda(\mathcal C(\mathbf G^{\otimes c})) = \lambda(\mathcal C(\mathbf G))^c.$
Hence, if $(\mathbf G,k)$ is a YES instance, then
\begin{align*}
    \lambda(\mathcal C(\mathbf G^{\otimes c}))\le k^c,
\end{align*}
so $(\mathbf G^{\otimes c},k^c)$ is a YES instance of $\gamma_0^c$-$\GapMDP_q$.
If $(\mathbf G,k)$ is a NO instance, then
\begin{align*}
    \lambda(\mathcal C(\mathbf G^{\otimes c}))
    >
    (\gamma_0 k)^c
    =
    \gamma_0^c k^c,
\end{align*}
so $(\mathbf G^{\otimes c},k^c)$ is a NO instance of $\gamma_0^c$-$\GapMDP_q$.
Because $c$ is a constant, this is a deterministic $\FPT$ many-one reduction.
Therefore, for every integer $c\ge 1$, the problem $\gamma_0^c$-$\GapMDP_q$ is $\W[1]$-hard under randomized $\FPT$ many-one reductions with one-sided error.

Now let $\gamma\ge 1$ be arbitrary.
Choose $c$ so that $\gamma_0^c\ge \gamma$.
Every YES instance of $\gamma_0^c$-$\GapMDP_q$ is also a YES instance of $\gamma$-$\GapMDP_q$, and every NO instance of $\gamma_0^c$-$\GapMDP_q$ is also a NO instance of $\gamma$-$\GapMDP_q$.
Hence the identity map is a deterministic $\FPT$ many-one reduction from $\gamma_0^c$-$\GapMDP_q$ to $\gamma$-$\GapMDP_q$.
Therefore $\gamma$-$\GapMDP_q$ is $\W[1]$-hard under randomized $\FPT$ many-one reductions with one-sided error.
\end{proof}
\subsection{W[1]-hardness of \texorpdfstring{$\gamma$-$\GapSVP_p$}{SVP} with One-Sided Error}
In this section, we prove a randomized $\FPT$ reduction with one-sided error from CVP to SVP using ideas similar to those for the hardness of MDP.

\begin{lemma}[Reed--Solomon codes, {\cite{journals/ReedS60}}]
\label{lem:rs-parity-check}
Let $q$ be a prime and let $m,r$ be integers with $1\le r\le m<q$.
Then one can construct in time $\rm{poly}(m,\log q)$ a matrix $\mathbf{P}\in\mathbb{F}_q^{r\times m}$ such that every nonzero vector $\vec{u}\in\mathbb{F}_q^m$ with $\|\vec{u}\|_0\le r$ satisfies $\mathbf{P}\vec{u}\neq 0$.
\end{lemma}

\begin{proof}
Choose pairwise distinct elements $a_1,\dots,a_m\in\mathbb{F}_q$ and define the Vandermonde matrix
\begin{align*}
\mathbf{P}:=
    \begin{pmatrix}
        1 & 1 & \cdots & 1\\
        a_1 & a_2 & \cdots & a_m\\
        \vdots & \vdots & & \vdots\\
        a_1^{r-1} & a_2^{r-1} & \cdots & a_m^{r-1}
        \end{pmatrix}\in\mathbb{F}_q^{r\times m}.
\end{align*}
This is the standard parity-check matrix of the $[m,m-r,r+1]_q$ Reed--Solomon code.
Suppose for contradiction that $\mathbf{P}\vec{u}=0$ for some nonzero $\vec{u}\in\mathbb{F}_q^m$ with $\|\vec{u}\|_0\le r$.
Let $S\subseteq[m]$ be the support of $\vec{u}$ and write $s:=|S|\le r$.
Restricting $\mathbf{P}$ to the columns indexed by $S$ yields an $r\times s$ Vandermonde matrix, and its first $s$ rows already form an invertible $s\times s$ Vandermonde matrix because the $a_i$'s are distinct.
Hence the columns indexed by $S$ are linearly independent, contradicting $\mathbf{P}\vec{u}=0$.
The construction is clearly deterministic and runs in time $\textrm{poly}(m,\log q)$.
\end{proof}

The following theorem gives a randomized algorithm for constructing locally dense lattices based on Reed--Solomon codes.
\begin{theorem}[{\cite[Theorem 4.3]{conf/stoc/BennettCGR23}}]
\label{thm:algo-ldl}
Fix real numbers $p \geq 1$ and $\gamma' \in [1, 2^{1/p})$.
Let
\begin{align*}
    \varepsilon &= (\gamma')^{-p} - \frac{1}{2} > 0, \;
    \gamma = \left\lceil \max\!\left(\frac{12}{\varepsilon}, \frac{1}{(1+\varepsilon/2)^{1/p} - 1} \right) \right\rceil, \;
    \alpha = \left(\frac{1}{(\gamma')^p} - \frac{2}{\gamma^p} \right)^{1/p}.
\end{align*}
Then there exists a randomized algorithm which, on input positive integers $m$ and $d$, runs in time $\textrm{poly}(m,d)$ and outputs with probability at least $0.99$ a $(p, \alpha, \gamma d, N, m', n')$-locally dense lattice $(\mathbf{A}, \vec{s})$, where
\begin{align*}
    m', n' &\leq \textrm{poly}(m,d), \;
    N = (2m(1+\gamma d))^{3(\gamma d)^p},
\end{align*}
provided that $m$ is sufficiently large compared to $p$.
Moreover, $\lambda_1^{(p)}(\mathcal{L}(\mathbf{A})) > \gamma d$ with probability $1$.
\end{theorem}

As in the case of MDP, by using a random matrix and the parity-check matrix of a Reed--Solomon code, we obtain a randomized $\FPT$ reduction with one-sided error.
\begin{theorem}
\label{thm:cvp-to-svp}
Fix a real number $p \geq 1$ and a rational constant $\gamma' \in [1,2^{1/p})$.
Define
\begin{align*}
\varepsilon &=(\gamma')^{-p}-\frac12>0 \; \text{ and } \;
\gamma =
\left\lceil\max\!\left(
\frac{12}{\varepsilon}, \frac{1}{(1+\varepsilon/2)^{1/p}-1}
\right)\right\rceil.
\end{align*}
Then there exists a randomized $\FPT$ many-one reduction with one-sided error from $\gamma$-$\GapCVP_p$ to $\gamma'$-$\GapSVP_p$.
\end{theorem}

\begin{proof}
Let $(\mathbf{B},\vec{t},d)$ be an instance of $\gamma$-$\GapCVP_p$ with $\mathbf{B}\in\Z^{m\times n}$, $\vec{t}\in\Z^m$, and $d\in \mathbb{Z}_{>0}$.

Let $m_{\mathrm{ldl}}(p)$ be a constant such that \Cref{thm:algo-ldl} applies whenever the first dimension is at least $m_{\mathrm{ldl}}(p)$.
If $m<m_{\mathrm{ldl}}(p)$, replace $(\mathbf B,\vec t,d)$ by
\begin{align*}
    \left(
    \begin{pmatrix}
        \mathbf B\\
        \mathbf 0_{(m_{\mathrm{ldl}}(p)-m)\times n}
    \end{pmatrix},
    \begin{pmatrix}
        \vec t\\
        \vec 0_{m_{\mathrm{ldl}}(p)-m}
    \end{pmatrix},
    d
    \right).
\end{align*}
For every $a\in\mathbb Z$ and every $\vec x\in\mathbb Z^n$,
\begin{align*}
    \left\|
    \begin{pmatrix}
        \mathbf B\vec x-a\vec t\\
        \vec 0
    \end{pmatrix}
    \right\|_p
    =
    \|\mathbf B\vec x-a\vec t\|_p.
\end{align*}
Hence, for every nonzero $a\in\mathbb Z$,
\begin{align*}
    \operatorname{dist}_p(\mathcal L(\mathbf B),a\vec t)
    =
    \operatorname{dist}_p\!\left(
        \mathcal L\!\left(
        \begin{pmatrix}
            \mathbf B\\
            \mathbf 0
        \end{pmatrix}
        \right),
        \begin{pmatrix}
            a\vec t\\
            \vec 0
        \end{pmatrix}
    \right),
\end{align*}
so this padding preserves both YES and NO instances of $\gamma$-$\GapCVP_p$ and keeps the parameter $d$ unchanged.
Therefore, without loss of generality, we may assume that $m\ge m_{\mathrm{ldl}}(p)$, i.e.\ that $m$ is sufficiently large compared to $p$.

Write $\gamma' = s_0/t_0$ in lowest terms with $s_0,t_0 \in \mathbb{Z}_{>0}$.
Set
\begin{align*}
\alpha &:=
\left(\frac{1}{(\gamma')^p}-\frac{2}{\gamma^p}\right)^{1/p}, \;
r := \min\{m,\lceil (\gamma d)^p\rceil\}, \;
L := s_0, \;
d' := \gamma t_0 d = L\frac{\gamma}{\gamma'}d \in \mathbb{Z}_{>0}.
\end{align*}
By \Cref{thm:algo-ldl}, in randomized time $\textrm{poly}(m,d)$ on input $(m,d)$ we obtain, with probability at least $0.99$, a $(p,\alpha,\gamma d,N,m',n')$-locally dense lattice $(\mathbf{A},\vec{s})$, where
\begin{align*}
N =(2m(1+\gamma d))^{3(\gamma d)^p}, \;
m',n' &\le \textrm{poly}(m,d), \; \text{ and }\;
\lambda_1^{(p)}(\mathcal{L}(\mathbf{A})) >\gamma d.
\end{align*}
Next, choose a prime $q$ satisfying
\begin{align*}
2\max\{m,\gamma d\}<q\le 4\max\{m,\gamma d\}.
\end{align*}
By Bertrand's postulate and primality testing, such a prime can be found deterministically in time $\textrm{poly}(m,d)$.
If $\lceil (\gamma d)^p\rceil\ge m$, let $\mathbf{P}:=I_m\in \mathbb{F}_q^{m\times m}$.
Otherwise, apply \Cref{lem:rs-parity-check} with parameters $(q,m,r)$ to obtain $\mathbf{P}\in\mathbb{F}_q^{r\times m}$ such that every nonzero vector $\vec{u}\in\mathbb{F}_q^m$ with $\lVert \vec{u} \rVert_0\le r$ satisfies $\mathbf{P}\vec{u}\neq 0$.
Let $\widehat{\mathbf{P}}\in\Z^{r\times m}$ be the canonical integer lift of $\mathbf{P}$ whose entries lie in $\{0,1,\dots,q-1\}$.
Define the reduction modulo $q$ maps
\begin{align*}
\pi_m &: \Z^m \to \mathbb{F}_q^m, & \pi_m(\vec{u}) &= \vec{u} \bmod q,\\
\pi_{m'} &: \Z^{m'} \to \mathbb{F}_q^{m'}, & \pi_{m'}(\vec{w}) &= \vec{w} \bmod q.
\end{align*}
Observe that every nonzero $\vec{u}\in\Z^m$ with $\lVert \vec{u} \rVert_p\le \gamma d$ satisfies $\mathbf{P}\pi_m(\vec{u}) \neq 0$.
Indeed, $|u_i|\le \lVert \vec{u} \rVert_p\le \gamma d<q/2$ for every coordinate, so $\pi_m(\vec{u})\neq 0$.
Moreover,
\begin{align*}
\lVert \pi_m(\vec{u}) \rVert_0
&\le \lVert \vec{u} \rVert_0
\le \lVert \vec{u} \rVert_p^p
\le (\gamma d)^p
\le r
\end{align*}
if $r<m$, while if $r=m$ then $\mathbf{P}=I_m$ and the claim is immediate.
Now sample $\mathbf{R}\in \mathbb{F}_q^{r\times m'}$ uniformly at random, let $\widehat{\mathbf{R}}\in\Z^{r\times m'}$ be its canonical integer lift.
More precisely, $\mathbf{R}$ is obtained by sampling the entries independently and uniformly at random from $\mathbb{F}_q$.
Define the lattice basis
\begin{align*}
\mathbf{B}'
:=
\begin{pmatrix}
\mathbf{B} & \mathbf{0}_{m \times n'} & -\vec{t} & \mathbf{0}_{m \times r} \\
\mathbf{0}_{m' \times n} & \mathbf{A} & -\vec{s} & \mathbf{0}_{m' \times r} \\
\mathbf{0}_{1 \times n} & \mathbf{0}_{1 \times n'} & 1 & \mathbf{0}_{1 \times r} \\
\beta \widehat{\mathbf{P}}\mathbf{B} & \beta \widehat{\mathbf{R}}\mathbf{A} & -\beta(\widehat{\mathbf{P}}\vec{t}+\widehat{\mathbf{R}}\vec{s}) & \beta q \mathbf{I}_r
\end{pmatrix}
\in \mathbb{Z}^{(m+m'+1+r) \times (n+n'+1+r)}\; \text{ with }\;
\beta := \gamma d + 1.
\end{align*}
And set
\begin{align*}
\widehat{\mathbf{B}} := L\mathbf{B}'.
\end{align*}
Output the lattice basis $\widehat{\mathbf{B}}$ together with the threshold $d'$.
Equivalently, for every $(\vec{x}, \vec{y}, a, \vec{z}) \in \Z^{n+n'+1+r}$, the corresponding lattice vector is
\begin{align*}
\vec{c}(\vec{x}, \vec{y}, a, \vec{z}) = \bigl(\vec{u}, \vec{w}, a, \beta(\widehat{\mathbf{P}}\vec{u}+\widehat{\mathbf{R}}\vec{w}+q\vec{z})\bigr),
\end{align*}
where $\vec{u} := \mathbf{B}\vec{x} - a\vec{t}$ and $\vec{w} := \mathbf{A}\vec{y} - a\vec{s}$.
This is the output instance of $\gamma'$-$\GapSVP_p$.

Since
\begin{align*}
m',n' \leq \textrm{poly}(m,d)
\qquad\text{and}\qquad
r \leq \lceil(\gamma d)^p\rceil,
\end{align*}
the output size and the running time are both at most $f(d)\cdot\textrm{poly}(|\mathbf{B}|+|\vec{t}|)$ for some computable function $f$.
The output parameter satisfies $d'=\gamma t_0d=O(d)$, because $\gamma$ and $t_0$ are constants that depend only on $\gamma'$.
Therefore, this is an $\FPT$ reduction, and the output parameter is $O(d)$.

It remains to prove completeness and soundness.

\paragraph*{Completeness.}
Assume that $(\mathbf{B}, \vec{t}, d)$ is a YES instance of $\gamma$-$\GapCVP_p$.
Then there exists $\vec{x}' \in \mathbb{Z}^n$ such that $\lVert \mathbf{B}\vec{x}' - \vec{t} \rVert_p \leq d$.
Let
\begin{align*}
\vec{u}' := \mathbf{B}\vec{x}' - \vec{t},
\qquad
\vec{a} := -\mathbf{P}\pi_m(\vec{u}') \in \mathbb{F}_q^r.
\end{align*}
Define
\begin{align*}
\mathcal{W}_s := (\mathcal{L}(\mathbf{A})-\vec{s})\cap \mathcal{B}_p^{m'}(\alpha \gamma d).
\end{align*}
Then, with probability at least $0.99$, we have $|\mathcal{W}_s|\ge N$.
Because every $\vec{w} \in \mathcal{W}_s$ satisfies
\begin{align*}
\lVert \vec{w} \rVert_p \leq \alpha \gamma d < \gamma d < q/2,
\end{align*}
all coordinates of $\vec{w}$ have absolute value $< q/2$.
Hence the reduction map $\pi_{m'}$ is injective on $\mathcal{W}_s$, and every $\pi_{m'}(\vec{w})$ is nonzero if $\vec{w} \neq 0$.
Choose a subset $\mathcal{T} \subseteq \mathcal{W}_s \setminus \{\vec{0}\}$ containing exactly one representative from each one-dimensional $\mathbb{F}_q$-subspace that intersects $\pi_{m'}(\mathcal{W}_s)$ nontrivially.
Then the set $\pi_{m'}(\mathcal{T})$ is pairwise linearly independent, and
\begin{align*}
    |\mathcal{T}| \geq \frac{|\mathcal{W}_s| - 1}{q - 1}.
\end{align*}
Since $|\mathcal{W}_s| \ge N$ and
\begin{align*}
    N=(2m(1+\gamma d))^{3(\gamma d)^p}
    \ge (2m(1+\gamma d))^3
    \ge 4\max\{m,\gamma d\}
    \ge q,
\end{align*}
we further obtain
\begin{align*}
    |\mathcal{T}|
    \ge \frac{|\mathcal{W}_s|-1}{q-1}
    \ge \frac{|\mathcal{W}_s|}{q}
    \ge \frac{N}{q}.
\end{align*}
For each $\vec{w} \in \mathcal{T}$, define the indicator random variable $X_{\vec{w}} := \mathbf{1}\bigl[\mathbf{R}\pi_{m'}(\vec{w}) = \vec{a}\bigr]$.
Since $\pi_{m'}(\vec{w}) \neq 0$, for uniformly random $\mathbf{R} \in \mathbb{F}_q^{r \times m'}$ we have
\begin{align*}
    \Pr[X_{\vec{w}} = 1] = q^{-r}.
\end{align*}
If $\vec{w}, \vec{w}' \in \mathcal{T}$ are distinct, then $\pi_{m'}(\vec{w})$ and $\pi_{m'}(\vec{w}')$ are linearly independent, so the pair $(X_{\vec{w}}, X_{\vec{w}'})$ is independent.
Thus the family $\{X_{\vec{w}}\}_{\vec{w} \in \mathcal{T}}$ is pairwise independent.
By \Cref{lem:chebyshev},
\begin{align*}
    \Pr[\forall \vec{w} \in \mathcal{T},\ X_{\vec{w}} = 0] \leq \frac{q^r}{|\mathcal{T}|} \leq \frac{q^{r+1}}{N}.
\end{align*}
Now let $M := 2m(1+\gamma d)$.
Since $q \leq 4\max\{m,\gamma d\} \leq 2M$ and $r \leq \lceil (\gamma d)^p \rceil \leq (\gamma d)^p + 1$, we obtain
\begin{align*}
\frac{q^{r+1}}{N}
\leq \frac{(2M)^{(\gamma d)^p + 2}}{M^{3(\gamma d)^p}}
= 4 \cdot 2^{(\gamma d)^p} \cdot M^{2-2(\gamma d)^p}.
\end{align*}
Because $\gamma'\in[1,2^{1/p})$, we have $\varepsilon=(\gamma')^{-p}-1/2\leq 1/2$, and therefore
\begin{align*}
\gamma\geq \left\lceil\frac{12}{\varepsilon}\right\rceil\geq 24.
\end{align*}
Hence $(\gamma d)^p \geq \gamma d \geq 24$ and $M=2m(1+\gamma d)\geq 50$.
Therefore,
\begin{align*}
\Pr[\forall \vec{w}\in \mathcal{T},\ X_{\vec{w}}=0]
\leq \frac{q^{r+1}}{N}
\leq 4\cdot 2^{24}\cdot 50^{-46}<0.01.
\end{align*}
So, conditioned on the event that $(\mathbf{A},\vec{s})$ is a $(p,\alpha,\gamma d,N,m',n')$-locally dense lattice, with probability at least $0.99$ over the choice of $\mathbf{R}$ there exists $\vec{w}\in \mathcal{T}\subseteq \mathcal{W}_s$ such that
\begin{align}
    \label{eq:congruence}
    \mathbf{R}\pi_{m'}(\vec{w})=-\mathbf{P}\pi_m(\vec{u}').
\end{align}
Fix such a vector $\vec{w}$.
Since $\vec{w}\in \mathcal{W}_s\subseteq \mathcal{L}(\mathbf{A})-\vec{s}$, there exists $\vec{y}\in\Z^{n'}$ such that $\vec{w} = \mathbf{A}\vec{y} - \vec{s}$ and $\lVert \vec{w} \rVert_p \leq \alpha \gamma d$.
Since $\widehat{\mathbf{P}}\vec{u}'+\widehat{\mathbf{R}}\vec{w} = 0\ (\mathrm{mod}\ q)$, this implies that there exists $\vec{z}\in\Z^r$ such that $\widehat{\mathbf{P}}\vec{u}'+\widehat{\mathbf{R}}\vec{w}+q\vec{z}=0$.
Hence $\vec{c}(\vec{x}', \vec{y}, 1, \vec{z})=(\vec{u}', \vec{w}, 1, \vec{0}_r)$.
Its $\ell_p$ norm satisfies
\begin{align*}
\lVert \vec{c}(\vec{x}', \vec{y}, 1, \vec{z}) \rVert_p^p
&= \lVert \vec{u}' \rVert_p^p + \lVert \vec{w} \rVert_p^p + 1 \\
&\le d^p + (\alpha \gamma d)^p + 1 \\
&\le 2d^p + \alpha^p\gamma^p d^p \\
&= \left(2+\alpha^p\gamma^p\right)d^p \\
&= \left(\frac{\gamma}{\gamma'}\right)^p d^p,
\end{align*}
where the last equality uses the definition of $\alpha$.
Hence $ \lVert \vec{c}(\vec{x}', \vec{y}, 1, \vec{z}) \rVert_p \le (\gamma/\gamma')d$.
Therefore the nonzero vector $L\vec{c}(\vec{x}', \vec{y}, 1, \vec{z}) \in \mathcal{L}(\widehat{\mathbf{B}})$ satisfies
\begin{align*}
\lVert L\vec{c}(\vec{x}', \vec{y}, 1, \vec{z}) \rVert_p
\le
L\frac{\gamma}{\gamma'}d
=
d'.
\end{align*}
Thus the output instance is a YES instance of $\gamma'$-$\GapSVP_p$.
Since the event that $(\mathbf{A},\vec{s})$ is $(p,\alpha,\gamma d,N,m',n')$-locally dense lattice has probability at least $0.99$ and, conditioned on that event, the probability over $\mathbf{R}$ that some $\vec{w}\in \mathcal{W}_s$ satisfies the congruence in \Cref{eq:congruence} is at least $0.99$, we conclude that
\begin{align*}
    \Pr[\text{output is a YES instance of } \gamma'\text{-}\GapSVP_p]\ge 0.99\cdot 0.99>\frac{2}{3}.
\end{align*}

\paragraph*{Soundness.}
Assume that $(\mathbf{B}, \vec{t}, d)$ is a NO instance of $\gamma$-$\GapCVP_p$.
We first show that, regardless of the choice of $\mathbf{R}$, every nonzero vector of $\mathcal{L}(\mathbf{B}')$ has $\ell_p$ norm strictly larger than $\gamma d$.
Let $\vec{c}=\vec{c}(\vec{x}, \vec{y}, a, \vec{z})\in \mathcal{L}(\mathbf{B}')\setminus\{0\}$ be an arbitrary nonzero lattice vector.
Suppose toward a contradiction that $\lVert \vec{c} \rVert_p\le \gamma d$.
Writing
\begin{align*}
\vec{c}=\bigl(\vec{u}, \vec{w}, a, \beta(\widehat{\mathbf{P}}\vec{u}+\widehat{\mathbf{R}}\vec{w}+q\vec{z})\bigr),
\end{align*}
where $\vec{u}=\mathbf{B}\vec{x}-a\vec{t}$ and $\vec{w}=\mathbf{A}\vec{y}-a\vec{s}$, we get
\begin{align*}
    \lVert \vec{u} \rVert_p^p+\lVert \vec{w} \rVert_p^p+|a|^p\le (\gamma d)^p.
\end{align*}
In particular, if $a\neq 0$, then by the NO promise for $\gamma$-$\GapCVP_p$,
\begin{align*}
    \lVert \vec{u} \rVert_p=\lVert \mathbf{B}\vec{x}-a\vec{t} \rVert_p>\gamma d,
\end{align*}
a contradiction.
Hence $a=0$.
If now $\vec{w}\neq 0$, then because $\lambda_1^{(p)}(\mathcal{L}(\mathbf{A}))>\gamma d$, we get $\lVert \vec{w} \rVert_p>\gamma d$, again a contradiction.
Therefore $a=0$ and $\vec{w}=0$.
So every short vector must be of the form
\begin{align*}
    \vec{c}(\vec{x}, \vec{y}, a, \vec{z})=(\vec{u}, 0, 0,       \beta(\widehat{\mathbf{P}}\vec{u}+q\vec{z})).
\end{align*}
If $\vec{u}=0$, then since $\vec{c} \neq 0$, $\vec{z}\neq 0$.
Hence
\begin{align*}
    \lVert \vec{c}(\vec{x}, \vec{y}, a, \vec{z}) \rVert_p\ge \beta q>\gamma d,
\end{align*}
a contradiction.
So we may assume $\vec{u}\neq 0$.
Since $\lVert \vec{u} \rVert_p\le \gamma d$, property of $\mathbf{P}$ implies $\mathbf{P}\pi_m(\vec{u})\neq 0$.
Equivalently,
\begin{align*}
\widehat{\mathbf{P}}\vec{u}+q\vec{z}\neq 0
\qquad\text{for every } \vec{z}\in\Z^r.
\end{align*}
Therefore the last block of $\vec{c}(\vec{x}, \vec{y}, a, \vec{z})$ is nonzero, so
\begin{align*}
\lVert \vec{c}(\vec{x}, \vec{y}, a, \vec{z}) \rVert_p\ge \beta>\gamma d,
\end{align*}
a contradiction.
We conclude that every nonzero vector of $\mathcal{L}(\mathbf{B}')$ has $\ell_p$ norm strictly larger than $\gamma d$.
Since $\mathcal{L}(\widehat{\mathbf{B}})=L\mathcal{L}(\mathbf{B}')$, every nonzero vector of $\mathcal{L}(\widehat{\mathbf{B}})$ has $\ell_p$ norm strictly larger than $L\gamma d = \gamma' d'$.
Therefore the output is a NO instance of $\gamma'$-$\GapSVP_p$.
This holds for every pair $(\mathbf{A},\vec{s})$ from \Cref{thm:algo-ldl} and every choice of $\mathbf{R}$, and thus the reduction has one-sided error.
\end{proof}

From the above, we immediately obtain \Cref{thm:SVP-is-W[1]hard}.
\begin{proof}[Proof of \Cref{thm:SVP-is-W[1]hard}]
Let $\gamma \in [1,2^{1/p})$ be the target approximation factor.
Choose any rational number $\eta$ such that $\gamma < \eta < 2^{1/p}$.
By \Cref{thm:CVP-is-hard,thm:cvp-to-svp}, the problem $\eta$-$\GapSVP_p$ is $\W[1]$-hard under randomized $\FPT$ many-one reductions with one-sided error.

Now observe that every YES instance of $\eta$-$\GapSVP_p$ is also a YES instance of $\gamma$-$\GapSVP_p$, and every NO instance of $\eta$-$\GapSVP_p$ is also a NO instance of $\gamma$-$\GapSVP_p$, because $\eta>\gamma$.
Hence any algorithm for $\gamma$-$\GapSVP_p$ would also solve $\eta$-$\GapSVP_p$ on the same input.

Therefore $\gamma$-$\GapSVP_p$ is $\W[1]$-hard under randomized $\FPT$ many-one reductions with one-sided error.
\end{proof}
\section{Deterministic Hardness under the Circuit Lower Bound}
\label{sec:derandomize}
In this section, we prove a general way to convert randomized hardness reductions with one-sided error into deterministic reductions under a circuit lower-bound assumption.
The argument applies when the target problem has an OR function, which allows several candidate outputs to be combined into a single instance.
We first present the statement for promise problems and then give its parameterized analogue.

\subsection{Conditional Hardness for Promise Problems}
We first treat ordinary promise problems.
The result shows that, under a circuit lower-bound assumption, randomized hardness reductions with one-sided error can be made deterministic when the target problem has a polynomial-time OR function.

\begin{theorem}
\label{thm:derandomize-np}
Assume that there exists an $\varepsilon>0$ such that
\begin{align*}
    \E \nsubseteq i.o.\text{-}\mathrm{NSIZE}(2^{\varepsilon n}).
\end{align*}
Let $\Pi=(\Pi_{\mathrm{YES}},\Pi_{\mathrm{NO}})$ and $\Pi'=(\Pi'_{\mathrm{YES}},\Pi'_{\mathrm{NO}})$ be promise problems.
Suppose that
\begin{enumerate}
    \item $\Pi$ is $\NP$-hard under deterministic polynomial-time many-one reductions;
    \item $\Pi'_{\mathrm{YES}}\in \NP$;
    \item $\Pi'$ has a polynomial-time OR function;
    \item there exists a randomized many-one reduction with one-sided error from $\Pi$ to $\Pi'$.
\end{enumerate}
Then $\Pi'$ is $\NP$-hard under deterministic polynomial-time many-one reductions.
\end{theorem}

For a fixed input $x$, consider the Boolean function that accepts exactly those random tapes for which the reduction outputs a YES instance of the target problem.
The next lemma shows that this function is computable by a polynomial-size nonadaptive $\SAT$-oracle circuit whenever the YES language of the target problem belongs to $\NP$.
\begin{lemma}
\label{lem:success-predicate-np}
Let $\Pi=(\Pi_{\mathrm{YES}},\Pi_{\mathrm{NO}})$ and $\Pi'=(\Pi'_{\mathrm{YES}},\Pi'_{\mathrm{NO}})$ be promise problems.
Assume that $\Pi'_{\mathrm{YES}}\in \NP$ and that there exists a randomized many-one reduction with one-sided error $R:\Pi \to \Pi'$.
Then there exists a fixed polynomial $P$ such that the following holds.

For every promised input $x$ of $\Pi$, let $n:=|x|$, let $\mu$ be the number of random bits used by $R$ on input $x$, let $m:=n+\mu$, and define
\begin{align*}
    \phi_x(r)=1 \iff R(x;r)\in \Pi'_{\mathrm{YES}},
\end{align*}
where only the first $\mu$ bits of $r$ are used as the random tape of $R$.
Then $\phi_x$ is computable by a nonadaptive $\SAT$-oracle circuit of size at most $P(m)$.
\end{lemma}

\begin{proof}
Fix a promised input $x$ of $\Pi$, and define $n,\mu,m$ as above.
Given $r\in\{0,1\}^m$, we can compute $R(x;r)$ deterministically in time polynomial in $n+\mu$, and hence in polynomial time in $m$.
Since $x$ is fixed and $n\le m$, hardwiring $x$ into the circuit increases the size by at most a polynomial in $m$.
Therefore there exists a polynomial-size deterministic circuit which computes $R(x;r)$ from $r$.

Since $\Pi'_{\mathrm{YES}}\in \NP$, there exists a polynomial-time many-one reduction
\begin{align*}
    f:\Pi'_{\mathrm{YES}}\le_m^p \SAT.
\end{align*}
Hence
\begin{align*}
    \phi_x(r)=1
    \iff R(x;r)\in \Pi'_{\mathrm{YES}}
    \iff f(R(x;r))\in \SAT.
\end{align*}
Thus $\phi_x$ can be computed by first deterministically computing $R(x;r)$, then deterministically computing $f(R(x;r))$, and finally making one oracle query to $\SAT$.
Since the resulting oracle circuit has at most one oracle gate on every input-to-output path, it is nonadaptive.
Therefore there exists a fixed polynomial $P$ such that $\phi_x$ is computable by a nonadaptive $\SAT$-oracle circuit of size at most $P(m)$.
\end{proof}

The previous lemma shows that, for each fixed input $x$, the predicate that checks whether a random tape makes the reduction output a YES instance is computable by a small nonadaptive $\SAT$-oracle circuit.
For a YES input $x$, one-sidedness guarantees that this predicate accepts a constant fraction of truly random tapes.
The next lemma shows that, under the circuit lower-bound assumption, the outputs of a pseudorandom generator still contain at least one accepting tape.
\begin{lemma}
\label{lem:hitting-seed-np}
Assume that there exists an $\varepsilon>0$ such that
\begin{align*}
    \E \nsubseteq i.o.\text{-}\mathrm{NSIZE}(2^{\varepsilon n}).
\end{align*}
Let $\Pi=(\Pi_{\mathrm{YES}},\Pi_{\mathrm{NO}})$ and $\Pi'=(\Pi'_{\mathrm{YES}},\Pi'_{\mathrm{NO}})$ be promise problems.
Assume that $\Pi'_{\mathrm{YES}}\in \NP$ and that there exists a randomized many-one reduction with one-sided error $R:\Pi \to \Pi'$.
Then there exist constants $c>1$, $a>1$, an integer $m_0$, and a family of functions
\begin{align*}
    G_m:\{0,1\}^{a\log m}\to\{0,1\}^m
    \qquad (m\ge m_0)
\end{align*}
such that the following hold.

First, there is a deterministic algorithm which, on input $m\ge m_0$ and $s\in\{0,1\}^{a\log m}$, computes $G_m(s)$ in time polynomial in $m$.
Second, for every $m\ge m_0$, the function $G_m$ is a $1/m^c$-PRG for $\mathrm{SIZE}^{\SAT}_{\|}(m^c)$.

Moreover, for every YES instance $x\in \Pi_{\mathrm{YES}}$, let $n:=|x|$, let $\mu$ be the number of random bits used by $R$ on input $x$, let $m:=n+\mu$, and define
\begin{align*}
    \phi_x(r)=1 \iff R(x;r)\in \Pi'_{\mathrm{YES}}.
\end{align*}
If $m\ge m_0$, then there exists a seed $s\in\{0,1\}^{a\log m}$ such that $\phi_x(G_m(s))=1$.
\end{lemma}

\begin{proof}
By \Cref{lem:success-predicate-np}, there exists a fixed polynomial $P$ such that, for every promised input $x$, the function $\phi_x$ is computable by a nonadaptive $\SAT$-oracle circuit of size at most $P(m)$.

Choose a constant $c>1$ such that $P(t)\le t^c$ for all sufficiently large $t$.
Apply \Cref{thm:prg-under-assumption} with this constant $c$.
Then there exist a constant $a>1$, an integer threshold $m_0$, and a family of functions
\begin{align*}
    G_m:\{0,1\}^{a\log m}\to\{0,1\}^m
    \qquad (m\ge m_0)
\end{align*}
such that this family is computable in time polynomial in $m$, and for every $m\ge m_0$, the function $G_m$ is a $1/m^c$-PRG for $\mathrm{SIZE}^{\SAT}_{\|}(m^c)$.
By increasing $m_0$ if necessary, assume also that $P(m)\le m^c$ and $1/m^c<2/3$ for every $m\ge m_0$.

Now fix a YES instance $x\in \Pi_{\mathrm{YES}}$, and define $n,\mu,m,\phi_x$ as in the statement.
Since $P(m)\le m^c$, the predicate $\phi_x$ belongs to $\mathrm{SIZE}^{\SAT}_{\|}(m^c)$.
Since $x\in \Pi_{\mathrm{YES}}$ and $R$ has one-sided error,
\begin{align*}
    \Pr_{r\leftarrow U_m}[\phi_x(r)=1]
    =
    \Pr_{u\leftarrow U_\mu}[R(x;u)\in \Pi'_{\mathrm{YES}}]
    \ge \frac23.
\end{align*}
Assume for contradiction that $\phi_x(G_m(s))=0$ for every seed $s\in\{0,1\}^{a\log m}$.
Then
\begin{align*}
    \Pr_{s\leftarrow U_{a\log m}}[\phi_x(G_m(s))=1]=0,
\end{align*}
and therefore
\begin{align*}
    \left|
        \Pr_{r\leftarrow U_m}[\phi_x(r)=1]
        -
        \Pr_{s\leftarrow U_{a\log m}}[\phi_x(G_m(s))=1]
    \right|
    \ge \frac23.
\end{align*}
For every $m\ge m_0$, we have $1/m^c<2/3$, contradicting the pseudorandomness of $G_m$.
Hence there exists a seed $s$ such that $\phi_x(G_m(s))=1$.
\end{proof}

We now turn the randomized reduction with one-sided error into a deterministic reduction.
For inputs with sufficiently large combined input-and-randomness length, we run the reduction on all outputs of the pseudorandom generator.
For smaller inputs, we simply enumerate all random tapes.
We then combine the resulting instances using the OR function of the target problem.
The previous lemma guarantees correctness on YES inputs, while one-sidedness guarantees correctness on NO inputs.
\begin{lemma}
\label{lem:derandomize-np-main}
Assume that there exists an $\varepsilon>0$ such that
\begin{align*}
    \E \nsubseteq i.o.\text{-}\mathrm{NSIZE}(2^{\varepsilon n}).
\end{align*}
Let $\Pi=(\Pi_{\mathrm{YES}},\Pi_{\mathrm{NO}})$ and $\Pi'=(\Pi'_{\mathrm{YES}},\Pi'_{\mathrm{NO}})$ be promise problems.
Suppose that
\begin{enumerate}
    \item $\Pi'_{\mathrm{YES}}\in \NP$;
    \item $\Pi'$ has a polynomial-time OR function;
    \item there exists a randomized many-one reduction with one-sided error $R:\Pi \to \Pi'$.
\end{enumerate}
Then there exists a deterministic polynomial-time many-one reduction $M:\Pi \to \Pi'$.
\end{lemma}

\begin{proof}
Let $C$ be a polynomial-time OR function for $\Pi'$.
By \Cref{lem:hitting-seed-np}, there exist constants $c>1$, $a>1$, an integer $m_0$, and a fixed family
\begin{align*}
    G_m:\{0,1\}^{a\log m}\to\{0,1\}^m
    \qquad (m\ge m_0)
\end{align*}
such that the conclusion of that lemma holds.

We define $M$ as follows.
On input $x$, let $n:=|x|$, let $\mu=\mu(x)$ be the number of random bits used by $R$ on input $x$, and let $m:=n+\mu$.

Define an ordered index set $S(x)$ and strings $r_s$ for $s\in S(x)$ as follows.
If $m<m_0$, let $S(x):=\{0,1\}^\mu$ and set $r_s:=s$ for every $s\in S(x)$.
If $m\ge m_0$, let $S(x):=\{0,1\}^{a\log m}$, and for each $s\in S(x)$ set $r_s:=G_m(s)$.

For each $s\in S(x)$, compute $y_s:=R(x;r_s)$, where only the first $\mu$ bits of $r_s$ are used as the random tape of $R$.
Finally, output
\begin{align*}
    M(x):=C\bigl(y_s \ :\ s\in S(x)\bigr),
\end{align*}
where the tuple is listed in lexicographic order of $S(x)$.

We verify that $M$ is computable in deterministic polynomial time.
Since $R$ is polynomial-time, $\mu(x)$ is bounded by a polynomial in $n$, and hence so is $m$.
If $m<m_0$, then $|S(x)|=2^\mu\le 2^{m_0}$, which is a constant.
If $m\ge m_0$, then $|S(x)|=2^{a\log m}=m^a$, which is polynomial in $n$.
Each call to $G_m$ (in the case $m\ge m_0$), each call to $R$, and the final call to $C$ all run in polynomial time.
Therefore $M$ is computable in deterministic polynomial time.

We next verify correctness.
Assume first that $x\in\Pi_{\mathrm{YES}}$.
If $m<m_0$, then, since $R$ has one-sided error,
\begin{align*}
    \Pr_{u\leftarrow U_\mu}[R(x;u)\in \Pi'_{\mathrm{YES}}]\ge \frac23.
\end{align*}
Hence there exists some $u^\star\in\{0,1\}^\mu=S(x)$ such that $R(x;u^\star)\in \Pi'_{\mathrm{YES}}$.
Therefore one of the inputs to $C$ belongs to $\Pi'_{\mathrm{YES}}$, and so $M(x)\in \Pi'_{\mathrm{YES}}$.

Assume now that $m\ge m_0$.
By \Cref{lem:hitting-seed-np}, there exists a seed $s^\star\in S(x)$ such that $y_{s^\star}=R(x;G_m(s^\star))\in \Pi'_{\mathrm{YES}}$.
Therefore, by the definition of polynomial-time OR function, $M(x)\in \Pi'_{\mathrm{YES}}$.

Assume next that $x\in\Pi_{\mathrm{NO}}$.
For each $s\in S(x)$, let $u_s$ denote the first $\mu$ bits of $r_s$.
By one-sidedness of $R$, for every random tape $u\in\{0,1\}^\mu$, $R(x;u)\in \Pi'_{\mathrm{NO}}$.
Hence, for every $s\in S(x)$, $y_s=R(x;u_s)\in \Pi'_{\mathrm{NO}}$.
Therefore, by the definition of polynomial-time OR function, $M(x)\in \Pi'_{\mathrm{NO}}$.

Thus $M$ is a deterministic polynomial-time many-one reduction from $\Pi$ to $\Pi'$.
\end{proof}

Thus, we obtain the desired claim.
\begin{proof}[Proof of \Cref{thm:derandomize-np}]
By \Cref{lem:derandomize-np-main}, there exists a deterministic polynomial-time many-one reduction $M:\Pi \to \Pi'$.
Since $\Pi$ is $\NP$-hard under deterministic polynomial-time many-one reductions, it follows that $\Pi'$ is $\NP$-hard under deterministic polynomial-time many-one reductions.
\end{proof}

\subsection{Conditional Hardness of Parameterized Promise Problems}
We next treat parameterized promise problems.
The result shows that, under the same circuit lower-bound assumption, randomized $\FPT$ hardness reductions with one-sided error can be made deterministic when the target problem has an FPT OR function and the random tapes leading to YES outputs can be recognized by small nonadaptive $\SAT$-oracle circuits.

\begin{theorem}
\label{thm:derandomize-fpt}
Assume that there exists an $\varepsilon>0$ such that
\begin{align*}
    \E \nsubseteq i.o.\text{-}\mathrm{NSIZE}(2^{\varepsilon n}).
\end{align*}
Let $\Pi=(\Pi_{\mathrm{YES}},\Pi_{\mathrm{NO}})$ and $\Pi'=(\Pi'_{\mathrm{YES}},\Pi'_{\mathrm{NO}})$ be parameterized promise problems.
Suppose that
\begin{enumerate}
    \item $\Pi$ is $\W[1]$-hard under deterministic $\FPT$ many-one reductions;
    \item $\Pi'$ has an FPT OR function;
    \item there exists a randomized $\FPT$ many-one reduction with one-sided error from $\Pi$ to $\Pi'$;
    \item the reduction in item~(3) has SAT-circuitizable success predicates.
\end{enumerate}
Then $\Pi'$ is $\W[1]$-hard under deterministic $\FPT$ many-one reductions.
\end{theorem}

For a fixed parameterized input $(x,k)$, consider the Boolean function that accepts exactly those random tapes for which the reduction outputs a YES instance of the target problem.
In the parameterized setting, we assume that this function can be computed by a polynomial-size nonadaptive $\SAT$-oracle circuit.
The next lemma records this assumption in the notation used below.
\begin{lemma}
\label{lem:success-predicate-fpt}
Let $\Pi=(\Pi_{\mathrm{YES}},\Pi_{\mathrm{NO}})$ and $\Pi'=(\Pi'_{\mathrm{YES}},\Pi'_{\mathrm{NO}})$ be parameterized promise problems.
Assume that there exists a randomized $\FPT$ many-one reduction with one-sided error $R:\Pi \to \Pi'$ and that $R$ has SAT-circuitizable success predicates.
Then there exists a fixed polynomial $P$ such that the following holds.

For every promised input $(x,k)$ of $\Pi$, let $n:=|\langle x,k\rangle|$, let $\mu$ be the number of random bits used by $R$ on input $(x,k)$, let $m:=n+\mu$, and define
\begin{align*}
    \phi_{x,k}(r)=1 \iff R((x,k);r)\in \Pi'_{\mathrm{YES}},
\end{align*}
where only the first $\mu$ bits of $r$ are used as the random tape of $R$.
Then $\phi_{x,k}$ is computable by a nonadaptive $\SAT$-oracle circuit of size at most $P(m)$.
\end{lemma}

\begin{proof}
This is immediate from the definition of SAT-circuitizable success predicates.
\end{proof}

For a YES input $(x,k)$, one-sidedness guarantees that a constant fraction of random tapes make the reduction output a YES instance.
Since the function recognizing such tapes is computable by a small nonadaptive $\SAT$-oracle circuit, the outputs of a pseudorandom generator contain at least one such tape under the circuit lower-bound assumption.
The next lemma states this for parameterized inputs.
\begin{lemma}
\label{lem:hitting-seed-fpt}
Assume that there exists an $\varepsilon>0$ such that
\begin{align*}
    \E \nsubseteq i.o.\text{-}\mathrm{NSIZE}(2^{\varepsilon n}).
\end{align*}
Let $\Pi=(\Pi_{\mathrm{YES}},\Pi_{\mathrm{NO}})$ and $\Pi'=(\Pi'_{\mathrm{YES}},\Pi'_{\mathrm{NO}})$ be parameterized promise problems.
Assume that there exists a randomized $\FPT$ many-one reduction with one-sided error $R:\Pi \to \Pi'$ and that $R$ has SAT-circuitizable success predicates.
Then there exist constants $c>1$, $a>1$, an integer $m_0$, and a family of functions
\begin{align*}
    G_m:\{0,1\}^{a\log m}\to\{0,1\}^m
    \qquad (m\ge m_0)
\end{align*}
such that the following hold.

First, there is a deterministic algorithm which, on input $m\ge m_0$ and $s\in\{0,1\}^{a\log m}$, computes $G_m(s)$ in time polynomial in $m$.
Second, for every $m\ge m_0$, the function $G_m$ is a $1/m^c$-PRG for $\mathrm{SIZE}^{\SAT}_{\|}(m^c)$.

Moreover, for every YES instance $(x,k)\in \Pi_{\mathrm{YES}}$, let $n:=|\langle x,k\rangle|$, let $\mu$ be the number of random bits used by $R$ on input $(x,k)$, let $m:=n+\mu$, and define
\begin{align*}
    \phi_{x,k}(r)=1 \iff R((x,k);r)\in \Pi'_{\mathrm{YES}}.
\end{align*}
If $m\ge m_0$, then there exists a seed $s\in\{0,1\}^{a\log m}$ such that $\phi_{x,k}(G_m(s))=1$.
\end{lemma}

\begin{proof}
By \Cref{lem:success-predicate-fpt}, there exists a fixed polynomial $P$ such that, for every promised input $(x,k)$, the function $\phi_{x,k}$ is computable by a nonadaptive $\SAT$-oracle circuit of size at most $P(m)$.

Choose a constant $c>1$ such that $P(t)\le t^c$ for all sufficiently large $t$.
Apply \Cref{thm:prg-under-assumption} with this constant $c$.
Then there exist a constant $a>1$, an integer threshold $m_0$, and a family of functions
\begin{align*}
    G_m:\{0,1\}^{a\log m}\to\{0,1\}^m
    \qquad (m\ge m_0)
\end{align*}
such that this family is computable in time polynomial in $m$, and for every $m\ge m_0$, the function $G_m$ is a $1/m^c$-PRG for $\mathrm{SIZE}^{\SAT}_{\|}(m^c)$.
By increasing $m_0$ if necessary, assume also that $P(m)\le m^c$ and $1/m^c<2/3$ for every $m\ge m_0$.

Now fix a YES instance $(x,k)\in \Pi_{\mathrm{YES}}$, and define $n,\mu,m,\phi_{x,k}$ as in the statement.
Since $P(m)\le m^c$, the predicate $\phi_{x,k}$ belongs to $\mathrm{SIZE}^{\SAT}_{\|}(m^c)$.
Since $(x,k)\in \Pi_{\mathrm{YES}}$ and $R$ has one-sided error,
\begin{align*}
    \Pr_{r\leftarrow U_m}[\phi_{x,k}(r)=1]
    =
    \Pr_{u\leftarrow U_\mu}[R((x,k);u)\in \Pi'_{\mathrm{YES}}]
    \ge \frac23.
\end{align*}
Assume for contradiction that $\phi_{x,k}(G_m(s))=0$ for every seed $s\in\{0,1\}^{a\log m}$.
Then
\begin{align*}
    \Pr_{s\leftarrow U_{a\log m}}[\phi_{x,k}(G_m(s))=1]=0,
\end{align*}
and hence
\begin{align*}
    \left|
        \Pr_{r\leftarrow U_m}[\phi_{x,k}(r)=1]
        -
        \Pr_{s\leftarrow U_{a\log m}}[\phi_{x,k}(G_m(s))=1]
    \right|
    \ge \frac23.
\end{align*}
For every $m\ge m_0$, we have $1/m^c<2/3$, contradicting the pseudorandomness of $G_m$.
Hence there exists a seed $s$ such that $\phi_{x,k}(G_m(s))=1$.
\end{proof}

We now turn the one-sided randomized $\FPT$ reduction into a deterministic $\FPT$ reduction.
For inputs with sufficiently large combined input-and-randomness length, we run the reduction on all outputs of the pseudorandom generator.
For smaller inputs, we simply enumerate all random tapes.
We then combine the resulting parameterized instances using the $\FPT$ OR function of the target problem.
The previous lemma guarantees correctness on YES inputs, one-sidedness guarantees correctness on NO inputs, and the parameter bounds in the reduction and the $\FPT$ OR function ensure that the final parameter depends only on the original parameter.
\begin{lemma}
\label{lem:derandomize-fpt-main}
Assume that there exists an $\varepsilon>0$ such that
\begin{align*}
    \E \nsubseteq i.o.\text{-}\mathrm{NSIZE}(2^{\varepsilon n}).
\end{align*}
Let $\Pi=(\Pi_{\mathrm{YES}},\Pi_{\mathrm{NO}})$ and $\Pi'=(\Pi'_{\mathrm{YES}},\Pi'_{\mathrm{NO}})$ be parameterized promise problems.
Suppose that
\begin{enumerate}
    \item $\Pi'$ has an FPT OR function;
    \item there exists a randomized $\FPT$ many-one reduction with one-sided error $R:\Pi \to \Pi'$;
    \item $R$ has SAT-circuitizable success predicates.
\end{enumerate}
Then there exists a deterministic $\FPT$ many-one reduction $M:\Pi \to \Pi'$.
\end{lemma}

\begin{proof}
Let $C$ be an FPT OR function for $\Pi'$.
By \Cref{lem:hitting-seed-fpt}, there exist constants $c>1$, $a>1$, an integer $m_0$, and a fixed family
\begin{align*}
    G_m:\{0,1\}^{a\log m}\to\{0,1\}^m
    \qquad (m\ge m_0)
\end{align*}
such that the conclusion of that lemma holds.

Fix computable functions $T_R,g_R$ and a constant $c_R>0$ witnessing the definition of randomized $\FPT$ many-one reduction with one-sided error for $R$.
Fix computable functions $T_{\mathrm{OR}},g_{\mathrm{OR}}$ and a constant $c_{\mathrm{OR}}>0$ witnessing the definition of $\FPT$ OR function for $C$.
Let $d_R:=\max\{1,\lceil c_R\rceil\}$.

We define $M$ as follows.
On input $(x,k)$, let $n:=|\langle x,k\rangle|$, let $\mu=\mu(x,k)$ be the number of random bits used by $R$ on input $(x,k)$, and let $m:=n+\mu$.

Define an ordered index set $S(x,k)$ and strings $r_s$ for $s\in S(x,k)$ as follows.
If $m<m_0$, let $S(x,k):=\{0,1\}^\mu$ and set $r_s:=s$ for every $s\in S(x,k)$.
If $m\ge m_0$, let $S(x,k):=\{0,1\}^{a\log m}$, and for each $s\in S(x,k)$ set $r_s:=G_m(s)$.

For each $s\in S(x,k)$, run $R$ on input $(x,k)$ using the first $\mu$ bits of $r_s$ as its random tape, obtaining $(y_s,\ell_s):=R((x,k);r_s)$.
Finally, output
\begin{align*}
    M(x,k):=C\bigl((y_s,\ell_s)\ :\ s\in S(x,k)\bigr),
\end{align*}
where the tuple is listed in lexicographic order of $S(x,k)$.

We verify that $M$ is a deterministic $\FPT$ many-one reduction.
Since $R$ runs in time at most $T_R(k)\cdot n^{c_R}$, the number of random bits $\mu$ satisfies
\begin{align*}
    \mu\le T_R(k)\cdot n^{c_R}\le T_R(k)\cdot (n+1)^{d_R}.
\end{align*}
Hence
\begin{align*}
    m=n+\mu\le (1+T_R(k))\cdot (n+1)^{d_R}.
\end{align*}

If $m<m_0$, then $|S(x,k)|=2^\mu\le 2^{m_0}$.
If $m\ge m_0$, then
\begin{align*}
    |S(x,k)|=2^{a\log m}=m^a
    \le (1+T_R(k))^a\cdot (n+1)^{ad_R}.
\end{align*}
Therefore, in all cases, $|S(x,k)|\le F_1(k)\cdot (n+1)^{ad_R}$ for some computable function $F_1$.

For every $s\in S(x,k)$, the output pair $(y_s,\ell_s)$ is produced within the running time of $R$, and hence $|\langle y_s,\ell_s\rangle| \le T_R(k)\cdot (n+1)^{d_R}$.
Moreover, every output parameter satisfies $\ell_s\le g_R(k)$.
Let
\begin{align*}
    N_z:=\sum_{s\in S(x,k)} |\langle y_s,\ell_s\rangle|.
\end{align*}
Then $N_z\le F_2(k)\cdot (n+1)^{(a+1)d_R}$ for some computable function $F_2$.
Hence $C$ is called on a tuple whose maximum parameter is at most $g_R(k)$, and by the definition of $\FPT$ OR function, the running time of $C$ is at most $T_{\mathrm{OR}}(g_R(k))\cdot N_z^{c_{\mathrm{OR}}}$.
Since each call to $G_m$ is polynomial in $m$, each call to $R$ takes time at most $T_R(k)\cdot n^{c_R}$, and the number of calls is at most $F_1(k)\cdot (n+1)^{ad_R}$, it follows that the overall running time of $M$ is of the form $F_3(k)\cdot (n+1)^{d}$ for some computable function $F_3$ and some constant $d>0$.
Equivalently, after changing the computable dependence on the parameter if necessary, this is a running-time bound of the form $F_3'(k)\cdot n^{d'}$ for some computable function $F_3'$ and some constant $d'>0$.
Thus $M$ is deterministic $\FPT$.
Its output parameter is bounded by $k_{\mathrm{out}}\le g_{\mathrm{OR}}(g_R(k))$, which depends only on $k$.

We next verify correctness.
Assume first that $(x,k)\in \Pi_{\mathrm{YES}}$.
If $m<m_0$, then, since $R$ has one-sided error,
\begin{align*}
    \Pr_{u\leftarrow U_\mu}[R((x,k);u)\in \Pi'_{\mathrm{YES}}]\ge \frac23.
\end{align*}
Hence there exists some $u^\star\in\{0,1\}^\mu=S(x,k)$ such that $R((x,k);u^\star)\in \Pi'_{\mathrm{YES}}$.
Therefore one of the inputs to $C$ belongs to $\Pi'_{\mathrm{YES}}$, and so $M(x,k)\in \Pi'_{\mathrm{YES}}$.

Assume now that $m\ge m_0$.
By \Cref{lem:hitting-seed-fpt}, there exists a seed $s^\star\in S(x,k)$ such that
\begin{align*}
    (y_{s^\star},\ell_{s^\star})=R((x,k);G_m(s^\star))\in \Pi'_{\mathrm{YES}}.
\end{align*}
Therefore, by the definition of $\FPT$ OR function, $M(x,k)\in \Pi'_{\mathrm{YES}}$.

Assume next that $(x,k)\in \Pi_{\mathrm{NO}}$.
For each $s\in S(x,k)$, let $u_s$ denote the first $\mu$ bits of $r_s$.
By one-sidedness of $R$, for every random tape $u\in\{0,1\}^\mu$, we have $R((x,k);u)\in \Pi'_{\mathrm{NO}}$.
Hence, for every $s\in S(x,k)$, $(y_s,\ell_s)=R((x,k);u_s)\in \Pi'_{\mathrm{NO}}$.
Therefore, by the definition of $\FPT$ OR function, $M(x,k)\in \Pi'_{\mathrm{NO}}$.

Thus $M$ is a deterministic $\FPT$ many-one reduction from $\Pi$ to $\Pi'$.
\end{proof}

Thus, we obtain the desired deterministic hardness statement for parameterized promise problems.
\begin{proof}[Proof of \Cref{thm:derandomize-fpt}]
By \Cref{lem:derandomize-fpt-main}, there exists a deterministic $\FPT$ many-one reduction $M:\Pi \to \Pi'$.
Since $\Pi$ is $\W[1]$-hard under deterministic $\FPT$ many-one reductions, it follows that $\Pi'$ is $\W[1]$-hard under deterministic $\FPT$ many-one reductions.
\end{proof}
\section{Conditional Hardness of MDP and SVP}
\label{sec:conditional}
In this section, we apply the deterministic hardness results from \Cref{sec:derandomize} to the coding and lattice problems studied in \Cref{sec:hardness}.
For each target problem, we first construct the required OR function and then verify that the randomized reduction used earlier satisfies the remaining condition needed to remove randomness.
\subsection{OR Functions of MDP}
We first consider $\gamma\text{-}\GapMDP_q$.
To apply the deterministic hardness result for parameterized promise problems, we need an FPT OR function for $\gamma\text{-}\GapMDP_q$.
The construction combines several generator matrices into one block-diagonal generator matrix after rescaling the distance thresholds to a common value.

The next lemma gives the required FPT OR function for $\gamma\text{-}\GapMDP_q$.
The common threshold is obtained by taking the least common multiple of the input thresholds, and each code is repeated so that its threshold becomes this common value.
\begin{lemma}[FPT OR function for $\gamma\text{-}\GapMDP_q$]
\label{lem:gapmdp-fpt-or-function}
Let $q\ge 2$ be a prime power, and let $\gamma \geq 1$ be a constant.
Then $\gamma\text{-}\GapMDP_q$, parameterized by the distance threshold $k$, has an FPT OR function.
\end{lemma}

\begin{proof}
We construct a deterministic algorithm $f_{\GapMDP}$ as follows.

On input a tuple of strings $w_1,\ldots,w_t$, first retain only those $w_i$ that are syntactically valid instances $(\mathbf G_i,k_i)$ of $\gamma\text{-}\GapMDP_q$ with $k_i\ge 1$.
If no input is retained, let $L_0:=\lfloor\gamma\rfloor+1$ and output the fixed NO instance $(\vec{1}_{L_0},1)$, where $\vec{1}_{L_0}$ is viewed as an $L_0\times 1$ generator matrix over $\mathbb{F}_q$.
Indeed, this one-dimensional code has minimum distance $L_0>\gamma$.

Let $\kappa:=\max_i k_i$ over the retained instances, and define $K:=\operatorname{lcm}(1,2,\ldots,\kappa)$.
For each retained instance $(\mathbf G_i,k_i)$, let $c_i:=K/k_i$, and let $\widetilde{\mathbf G}_i$ be the matrix obtained from $\mathbf G_i$ by repeating each row exactly $c_i$ times.
Finally output
\begin{align*}
    \bigl(\mathbf G_{\mathrm{out}},k_{\mathrm{out}}\bigr)
    :=
    \left(
        \begin{pmatrix}
            \widetilde{\mathbf G}_1 &        &        & \mathbf 0 \\
                                    & \widetilde{\mathbf G}_2 &        & \\
                                    &        & \ddots & \\
            \mathbf 0               &        &        & \widetilde{\mathbf G}_s
        \end{pmatrix},
        \,K
    \right),
\end{align*}
where $s$ is the number of retained instances.

We verify correctness.
For each $i$, every codeword of $\widetilde{\mathbf G}_i$ is obtained from a codeword of $\mathbf G_i$ by repeating each coordinate $c_i$ times, and hence $\lambda(\mathcal{C}(\widetilde{\mathbf{G}}_i))=c_i\,\lambda(\mathcal{C}(\mathbf{G}_i))$.
Since $\mathbf{G}_{\mathrm{out}}$ is block diagonal,
\begin{align*}
    \lambda(\mathcal{C}(\mathbf{G}_{\mathrm{out}}))
    =
    \min_{i\in[s]} \lambda(\mathcal{C}(\widetilde{\mathbf G}_i))
    =
    \min_{i\in[s]} c_i\,\lambda(\mathcal{C}(\mathbf G_i)).
\end{align*}

If one of the original inputs is a YES instance, say $(\mathbf G_j,k_j) \textrm{ is a YES instance of } \gamma\text{-}\GapMDP_q$, then it is retained and $\lambda(\mathcal{C}(\mathbf G_j))\le k_j$.
Therefore
\begin{align*}
\lambda(\mathcal{C}(\widetilde{\mathbf G}_j))\le c_jk_j=K,
\end{align*}
so $(\mathbf{G}_{\mathrm{out}},K) \textrm{ is a YES instance of } \gamma\text{-}\GapMDP_q$.

If all original inputs are NO instances, then every retained instance is NO, and hence $\lambda(\mathcal{C}(\mathbf G_i))>\gamma k_i$ for every retained $i$.
Thus
\begin{align*}
    \lambda(\mathcal{C}(\widetilde{\mathbf{G}}_i))> \gamma c_i k_i=\gamma K
\end{align*}
for every retained $i$, and so $\lambda(\mathcal{C}(\mathbf G_{\mathrm{out}}))>\gamma K$.
Hence $(\mathbf G_{\mathrm{out}},K) \textrm{ is a NO instance of } \gamma\text{-}\GapMDP_q$.

We now verify the FPT requirements.
The output threshold is $k_{\mathrm{out}}=K$, which depends only on $\kappa$.
Moreover, the algorithm runs in time $T(\kappa)\cdot N^{O(1)}$, where $N$ is the total input length and $T(\kappa)$ accounts for computing $K$ and explicitly writing down the repeated rows.
Therefore $f_{\GapMDP}$ is an FPT OR function for $\gamma\text{-}\GapMDP_q$.
\end{proof}

We now combine the randomized hardness of $\gamma\text{-}\GapMDP_q$ with the FPT OR function constructed above.
The only remaining point is to verify the circuit condition for the randomized reduction.

\begin{proof}[Proof of \Cref{thm:MDP-is-W[1]hard-conditional}]
Fix a prime power $q\ge 2$ and a constant $\gamma\ge 1$.
By \Cref{thm:MDP-is-W[1]hard}, $\gamma\text{-}\GapMDP_q$ is $\W[1]$-hard under randomized $\FPT$ many-one reductions with one-sided error.
By \Cref{lem:gapmdp-fpt-or-function}, $\gamma\text{-}\GapMDP_q$ has an FPT OR function.
More concretely, in the application of \Cref{thm:derandomize-fpt}, we take the source problem to be $4q$-\GapNCP$_q$ and take $R$ to be the concrete reduction obtained in the proof of \Cref{thm:MDP-is-W[1]hard}, including the constant number of tensorings needed for the target approximation factor.
By \Cref{thm:NCP-is-hard}, this source problem is $\W[1]$-hard under deterministic $\FPT$ many-one reductions.

It remains to verify the SAT-circuitizable-success-predicate condition in \Cref{thm:derandomize-fpt} for this concrete reduction $R$.
Fix an input $(x,k)$ of the source problem.
Let $n:=|\langle x,k\rangle|$, let $\mu$ be the number of random bits used by $R$ on $(x,k)$, let $m:=n+\mu$, and define
\begin{align*}
    \phi_{x,k}(r)=1 \iff R((x,k);r) \textrm{ is a YES instance of } \gamma\text{-}\GapMDP_q,
\end{align*}
where only the first $\mu$ bits of $r$ are used as the random tape of $R$.
Since $q$ is fixed, the YES language of $\gamma\text{-}\GapMDP_q$ belongs to $\NP$: for an instance $(\mathbf G,\ell)$ with $\mathbf G\in\mathbb F_q^{m'\times n'}$, a witness is a vector $\vec z\in\mathbb F_q^{n'}$ such that
\begin{align*}
    \mathbf G\vec z\neq \vec 0
    \qquad\text{and}\qquad
    \|\mathbf G\vec z\|_0\le \ell.
\end{align*}

Moreover, since $R$ is a randomized $\FPT$ many-one reduction, there exist a computable function $T_R$ and a constant $c_R>0$ such that, on every input $(x,k)$ of length $n$, the running time of $R$ is at most $T_R(k)n^{c_R}$.
By padding the random tape if necessary, we may assume $\mu\ge T_R(k)n^{c_R}$.
Therefore, given $r\in\{0,1\}^m$, we can compute $R((x,k);r)$ deterministically in time $O(\mu)=\poly(m)$, and after hardwiring $(x,k)$ this yields a deterministic circuit of size $\poly(m)$ that computes $R((x,k);r)$ from $r$.

Composing this circuit with a polynomial-time many-one reduction from $(\gamma\text{-}\GapMDP_q)_{\mathrm{YES}}$ to $\SAT$ gives a polynomial-size nonadaptive $\SAT$-oracle circuit for $\phi_{x,k}$.
Thus $R$ has SAT-circuitizable success predicates.

All hypotheses of \Cref{thm:derandomize-fpt} are therefore satisfied with $\Pi'=\gamma\text{-}\GapMDP_q$.
Applying \Cref{thm:derandomize-fpt}, we conclude that $\gamma\text{-}\GapMDP_q$ is $\W[1]$-hard under deterministic $\FPT$ many-one reductions.
\end{proof}
\subsection{OR Functions of SVP}
We next turn to $\gamma\text{-}\GapSVP_p$.
We need two OR functions: a polynomial-time OR function for the ordinary promise-problem setting and an FPT OR function for the parameterized setting.
Both constructions use block-diagonal sums of the input bases.

For the ordinary promise-problem setting, the output threshold must be computable in polynomial time in the total input length.
The next lemma achieves this by scaling the input bases so that all thresholds become a common value, and then taking their block-diagonal sum.
\begin{lemma}[Polynomial-time OR function for $\gamma\text{-}\GapSVP_p$]
\label{lem:gapsvp-poly-or-function}
Let $p\in[1,\infty]$ and let $\gamma \geq 1$ be a constant.
Then $\gamma\text{-}\GapSVP_p$ has a polynomial-time OR function.
\end{lemma}

\begin{proof}
We construct a deterministic polynomial-time algorithm $f_{\GapSVP}$ as follows.

On input a tuple of strings $w_1,\ldots,w_t$, first retain only those $w_i$ that are syntactically valid instances $(\mathbf B_i,d_i)$ of $\gamma\text{-}\GapSVP_p$ with $d_i\ge 1$.
If no input is retained, output the fixed NO instance $([\lfloor\gamma\rfloor+2],1)$.

Let
\begin{align*}
    D := \prod_{i=1}^s d_i,
\end{align*}
where $s$ is the number of retained instances.
For each retained instance, let
\begin{align*}
    c_i := D/d_i = \prod_{j\neq i} d_j \in \mathbb Z_{>0},
    \qquad
    \widetilde{\mathbf B}_i := c_i \mathbf B_i,
\end{align*}
and output
\begin{align*}
    \bigl(\mathbf B_{\mathrm{out}},d_{\mathrm{out}}\bigr)
    :=
    \left(
        \begin{pmatrix}
            \widetilde{\mathbf B}_1 &        &        & \mathbf 0 \\
                                    & \widetilde{\mathbf B}_2 &        & \\
                                    &        & \ddots & \\
            \mathbf 0               &        &        & \widetilde{\mathbf B}_s
        \end{pmatrix},
        \,D
    \right).
\end{align*}

We verify correctness.
For each $i$, $\lambda_1^{(p)}(\mathcal L(\widetilde{\mathbf B}_i))  = c_i\,\lambda_1^{(p)}(\mathcal L(\mathbf B_i))$.
Since $\mathbf B_{\mathrm{out}}$ is block diagonal, $\lambda_1^{(p)}(\mathcal L(\mathbf B_{\mathrm{out}})) = \min_{i\in[s]} c_i\,\lambda_1^{(p)}(\mathcal L(\mathbf B_i))$.

If one of the original inputs is a YES instance, say $(\mathbf B_j,d_j)$ is a YES instance of $\gamma\text{-}\GapSVP_p$, then
\begin{align*}
    \lambda_1^{(p)}(\mathcal L(\widetilde{\mathbf B}_j))
    \le c_j d_j = D,
\end{align*}
and hence $(\mathbf B_{\mathrm{out}},D)$ is a YES instance of $\gamma\text{-}\GapSVP_p$.

If all original inputs are NO instances, then $\lambda_1^{(p)}(\mathcal L(\mathbf B_i))>\gamma d_i$ for every retained $i$, and so
\begin{align*}
    \lambda_1^{(p)}(\mathcal L(\widetilde{\mathbf B}_i))
    >
    \gamma c_i d_i
    =
    \gamma D
\end{align*}
for every retained $i$.
Therefore $(\mathbf B_{\mathrm{out}},D)$ is a NO instance of $\gamma\text{-}\GapSVP_p$.

Finally, the bit length of $D$ is at most the sum of the bit lengths of the retained thresholds $d_i$, so the entire construction runs in deterministic polynomial time in the total input length.
Thus $f_{\GapSVP}$ is a polynomial-time OR function for $\gamma\text{-}\GapSVP_p$.
\end{proof}

We now prove the conditional deterministic $\NP$-hardness of $\gamma\text{-}\GapSVP_2$.
The randomized hardness result and the polynomial-time OR function above allow us to apply the deterministic hardness result for promise problems.
\begin{proof}[Proof of \Cref{thm:SVP-is-NPhard-conditional}]
Fix a constant $\gamma\ge 1$.
Let
\begin{align*}
    \Pi_{\SAT}:=(\SAT,\overline{\SAT})
\end{align*}
be the promise problem corresponding to the language $\SAT$.

By \Cref{thm:SVP-is-NPhard}, $\gamma\text{-}\GapSVP_2$ is $\NP$-hard under randomized many-one reductions with one-sided error.
Applying this statement to the NP-complete language $\SAT$, there exists a randomized polynomial-time many-one reduction with one-sided error
\begin{align*}
    R:\Pi_{\SAT}\to \gamma\text{-}\GapSVP_2 .
\end{align*}
Moreover, $\Pi_{\SAT}$ is $\NP$-hard under deterministic polynomial-time many-one reductions.

The YES language of $\gamma\text{-}\GapSVP_2$ belongs to $\NP$.
Indeed, a certificate for a YES instance $(\mathbf B,d)$ is a nonzero integer vector $\vec z$ such that $\|\mathbf B\vec z\|_2\le d$, or equivalently $\sum_i (\mathbf B\vec z)_i^2 \le d^2$.
A standard Cramer's-rule bound shows that, whenever such a certificate exists, one exists with polynomial bit length.
The inequality above can then be checked in deterministic polynomial time.

By \Cref{lem:gapsvp-poly-or-function}, $\gamma\text{-}\GapSVP_2$ has a polynomial-time OR function.
Therefore all hypotheses of \Cref{thm:derandomize-np} are satisfied with
\begin{align*}
    \Pi=\Pi_{\SAT},
    \qquad
    \Pi'=\gamma\text{-}\GapSVP_2 .
\end{align*}
Applying \Cref{thm:derandomize-np}, we obtain a deterministic polynomial-time many-one reduction from $\Pi_{\SAT}$ to $\gamma\text{-}\GapSVP_2$.
Since $\Pi_{\SAT}$ is $\NP$-hard under deterministic polynomial-time many-one reductions, $\gamma\text{-}\GapSVP_2$ is $\NP$-hard under deterministic polynomial-time many-one reductions.
\end{proof}

For the parameterized hardness result, we also need an FPT OR function.
Here the thresholds are positive integers, so we can use the least common multiple of the input thresholds as the common output threshold.
\begin{lemma}[FPT OR function for $\gamma\text{-}\GapSVP_p$]
\label{lem:gapsvp-fpt-or-function}
Let $p\in[1,\infty]$ and let $\gamma \geq 1$ be a constant.
Then $\gamma\text{-}\GapSVP_p$, parameterized by the threshold $d\in\mathbb N$, has an FPT OR function.
\end{lemma}

\begin{proof}
We use the same block-diagonal construction as in \Cref{lem:gapsvp-poly-or-function}, but now all retained thresholds are positive integers.

On input a tuple of strings $w_1,\ldots,w_t$, first retain only those $w_i$ that are syntactically valid instances $(\mathbf B_i,d_i)$ of $\gamma\text{-}\GapSVP_p$ with $d_i\ge 1$.
If no input is retained, output the fixed NO instance $([\lfloor \gamma \rfloor + 2], 1)$.

Let $\kappa:=\max_i d_i$ over the retained instances, and define $D:=\operatorname{lcm}(1,2,\ldots,\kappa)$.
For each retained instance, let $c_i:=D/d_i$, define $\widetilde{\mathbf B}_i:=c_i\mathbf B_i$, and output
\begin{align*}
    \bigl(\mathbf B_{\mathrm{out}},d_{\mathrm{out}}\bigr)
    :=
    \left(
        \begin{pmatrix}
            \widetilde{\mathbf B}_1 &        &        & \mathbf 0 \\
                                    & \widetilde{\mathbf B}_2 &        & \\
                                    &        & \ddots & \\
            \mathbf 0               &        &        & \widetilde{\mathbf B}_s
        \end{pmatrix},
        \,D
    \right).
\end{align*}

Exactly as above, $\lambda_1^{(p)}(\mathcal L(\mathbf B_{\mathrm{out}})) = \min_{i\in[s]} c_i\,\lambda_1^{(p)}(\mathcal L(\mathbf B_i))$.

If one of the original inputs is a YES instance, say $(\mathbf B_j,d_j) \textrm{ is a YES instance of } \gamma\text{-}\GapSVP_p$, then
\begin{align*}
    \lambda_1^{(p)}(\mathcal L(\widetilde{\mathbf B}_j))
    \le c_j d_j = D,
\end{align*}
and hence $(\mathbf B_{\mathrm{out}},D) \textrm{ is a YES instance of } \gamma\text{-}\GapSVP_p$.

If all original inputs are NO instances, then $\lambda_1^{(p)}(\mathcal L(\mathbf B_i))>\gamma d_i$ for every retained $i$, and so
\begin{align*}
    \lambda_1^{(p)}(\mathcal L(\widetilde{\mathbf B}_i))
    >
    \gamma c_i d_i
    =
    \gamma D
\end{align*}
for every retained $i$.
Therefore $(\mathbf B_{\mathrm{out}},D) \textrm{ is a NO instance of } \gamma\text{-}\GapSVP_p$.

The output threshold is $d_{\mathrm{out}}=D$, which depends only on $\kappa$.
The construction runs in time $T(\kappa)\cdot N^{O(1)}$, where $N$ is the total input length and $T(\kappa)$ accounts for computing $D$ and scaling the basis matrices.
Therefore this is an FPT OR function for $\gamma\text{-}\GapSVP_p$.
\end{proof}

Finally, we prove the conditional deterministic $\W[1]$-hardness of $\gamma\text{-}\GapSVP_p$.
The randomized $\FPT$ hardness reduction was proved earlier, and the previous lemma provides the required FPT OR function.
As in the MDP case, it remains to verify that the random tapes producing YES outputs can be recognized by a small nonadaptive $\SAT$-oracle circuit.

\begin{proof}[Proof of \Cref{thm:SVP-is-W[1]hard-conditional}]
Fix an integer $p\ge 1$, and let $\gamma\in[1,2^{1/p})$ be a constant.
Choose a rational number $\eta$ such that $\gamma < \eta < 2^{1/p}$.
It suffices to prove that $\eta$-$\GapSVP_p$ is $\W[1]$-hard under deterministic $\FPT$ many-one reductions, because any algorithm for $\gamma$-$\GapSVP_p$ also solves $\eta$-$\GapSVP_p$ on the same promised instances.

By \Cref{thm:SVP-is-W[1]hard}, $\eta$-$\GapSVP_p$ is $\W[1]$-hard under randomized $\FPT$ many-one reductions with one-sided error.
By \Cref{lem:gapsvp-fpt-or-function}, $\eta$-$\GapSVP_p$ has an FPT OR function.
More concretely, let $\Gamma$ be the source approximation factor in \Cref{thm:cvp-to-svp} corresponding to the rational target factor $\eta$.
In the application of \Cref{thm:derandomize-fpt}, we take the source problem to be $\Gamma$-\GapCVP$_p$ and take $R$ to be the concrete reduction from \Cref{thm:cvp-to-svp}.
By \Cref{thm:CVP-is-hard}, this source problem is $\W[1]$-hard under deterministic $\FPT$ many-one reductions.

It remains to verify the SAT-circuitizable-success-predicate condition in \Cref{thm:derandomize-fpt} for this concrete reduction $R$.
Fix an input $(x,k)$ of the source problem.
Let $n:=|\langle x,k\rangle|$, let $\mu$ be the number of random bits used by $R$ on $(x,k)$, let $m:=n+\mu$, and define
\begin{align*}
    \phi_{x,k}(r)=1 \iff R((x,k);r) \textrm{ is a YES instance of } \eta\text{-}\GapSVP_p,
\end{align*}
where only the first $\mu$ bits of $r$ are used as the random tape of $R$.

Since $p$ is a fixed integer, the YES language of $\eta$-$\GapSVP_p$ belongs to $\NP$.
Indeed, the YES language does not depend on the approximation factor and is exactly
\begin{align*}
    \{(\mathbf B,d) : \mathbf B \text{ is a lattice basis and } d\in\mathbb Z_{>0},\ \lambda_1^{(p)}(\mathcal L(\mathbf B)) \le d\}.
\end{align*}
Hence there exists a polynomial-time many-one reduction
\begin{align*}
    f:(\eta\text{-}\GapSVP_p)_{\mathrm{YES}}\le_m^p \SAT .
\end{align*}

Moreover, since $R$ is a randomized $\FPT$ many-one reduction, there exist a computable function $T_R$ and a constant $c_R>0$ such that on every input $(x,k)$ of length $n$ the running time of $R$ is at most $T_R(k)\cdot n^{c_R}$.
By padding the random tape (if necessary), we may assume that on input $(x,k)$ the algorithm reads at least $T_R(k)\cdot n^{c_R}$ random bits; in particular $\mu \ge T_R(k)\cdot n^{c_R}$.
Therefore, given $r\in\{0,1\}^m$ we can compute $R((x,k);r)$ deterministically in time $O(\mu)=\textrm{poly}(m)$, and after hardwiring $(x,k)$ this yields a deterministic circuit of size $\poly(m)$ that computes $R((x,k);r)$ from $r$.

Consequently,
\begin{align*}
\phi_{x,k}(r)=1 \iff R((x,k);r)\textrm{ is a YES instance of }\eta\text{-}\GapSVP_p
           \iff f(R((x,k);r))\in \SAT.
\end{align*}
Thus $\phi_{x,k}$ is computable by a polynomial-size nonadaptive $\SAT$-oracle circuit: first compute $R((x,k);r)$, then compute $f(\cdot)$, and finally make one $\SAT$ query.
Hence $R$ has SAT-circuitizable success predicates.

All hypotheses of \Cref{thm:derandomize-fpt} are therefore satisfied with $\Pi'=\eta\text{-}\GapSVP_p$.
Applying \Cref{thm:derandomize-fpt}, we conclude that $\eta\text{-}\GapSVP_p$ is $\W[1]$-hard under deterministic $\FPT$ many-one reductions.
By the first paragraph, $\gamma\text{-}\GapSVP_p$ is also $\W[1]$-hard under deterministic $\FPT$ many-one reductions.
\end{proof}

\printbibliography
\end{document}